\documentclass[letterpaper,10pt]{article}
\usepackage{csquotes}
\usepackage{todonotes}
\usepackage{natbib}
\usepackage{hhline}
\usepackage{authblk}
\usepackage[english]{babel}

\usepackage[T1]{fontenc}
\usepackage[tt=false, type1=true]{libertine}
\usepackage[top=1in, bottom=1in, left=1in, right=1in]{geometry}

\usepackage{amsmath}
\usepackage{graphicx}
\usepackage{amssymb}
\usepackage{amsthm}
\usepackage{graphicx}
\usepackage{xcolor}
\usepackage{amsfonts} 
\usepackage{bm}
\usepackage{multirow} 
\usepackage{authblk}
\usepackage{thm-restate}
\usepackage{booktabs}
\usepackage{colortbl}
\usepackage[bookmarks=true,hypertexnames=false,pagebackref,colorlinks=true,citecolor=blue, linkcolor=black, urlcolor=violet]{hyperref} 

\usepackage{tikz}
\usepackage[linesnumbered,ruled,vlined]{algorithm2e}
\usepackage{color}
\usepackage{xspace}
\usepackage{framed}
\usepackage{cleveref}
\usepackage{enumitem}
\usepackage{macros}

\RequirePackage{bm}

\theoremstyle{definition}
\newtheorem{definition}{Definition}

\theoremstyle{plain}
\newtheorem{theorem}{Theorem}
\newtheorem{lemma}{Lemma}

\newtheorem{corollary}{Corollary}
\newtheorem{proposition}{Proposition}
\newtheorem{observation}[theorem]{Observation}
\newtheorem{takeaway}{Takeaway}

\newtheorem{example}{Example}
\newtheorem{problem}{Problem}

\title{\bf Human-AI Collaboration with Misaligned Preferences}

\author[1]{Jiaxin Song}
\author[2]{Parnian Shahkar}
\author[1,3]{Kate Donahue}
\author[1]{Bhaskar Ray Chaudhury}

\affil[1]{University of Illinois, Urbana-Champaign}
\affil[2]{University of California, Irvine}
\affil[3]{Massachusetts Institute of Technology}
\affil[ ]{\{jiaxins8,kpd,braycha\}@illinois.edu, shahkarp@uci.edu}

\date{}

\begin{document}
\maketitle
\begin{abstract}
In many real-life settings, algorithms play the role of assistants, while humans ultimately make the final decision. Often, algorithms specifically act as curators, narrowing down a wide range of options into a smaller subset that the human picks between: consider content recommendation or chatbot responses to questions with multiple valid answers. Crucially, humans may not know their own preferences perfectly either, but instead may only have access to a noisy sampling over preferences. Algorithms can assist humans by curating a smaller subset of items, but must also face the challenge of \emph{misalignment}: humans may have different preferences from each other (and from the algorithm), and the algorithm may not know the exact preferences of the human they are facing at any point in time. In this paper, we model and theoretically study such a setting. Specifically, we show instances where humans benefit by collaborating with a misaligned algorithm. Surprisingly, we show that humans gain more utility from a misaligned algorithm (which makes different mistakes) than from an aligned algorithm. Next, we build on this result by studying what properties of algorithms maximize human welfare, when the goals could be either utilitarian welfare or ensuring all humans benefit. We conclude by discussing implications for designers of algorithmic tools and policymakers. 
\end{abstract}


\section{Introduction}
In recent years, advances in artificial intelligence and machine learning have become increasingly integrated into our daily lives: algorithms help suggest movies for us to watch, routes for us to drive on, or even generate novel content for us to rely on. However, outputs of algorithmic tools in this context almost never have the \enquote{final say}: for example, while an algorithm can suggest multiple movies, the human makes the final decision on which movie she ultimately wants to watch. As a result, we often care about studying the performance of the human-algorithm \emph{system}, rather than the performance of the algorithm in isolation. Human-algorithm systems have been studied extensively in more applied contexts, and in recent years, a more formal theoretical study of human-algorithm systems has grown (see Section \ref{sec:related} for more discussion). 

In this paper, we focus on a specific instance of human-algorithm collaboration: where the algorithm acts as a \emph{curator}, and the humans are \emph{noisy} and potentially \emph{misaligned} with each other and with the algorithm. Informally, an algorithm acts as a \emph{curator} when its role is to narrow down a larger set of items to a smaller set, among which the human picks her favorite. This type of role is one of the most common in human-algorithm systems: e.g., in content recommendation, search, and some types of categorical prediction\footnote{For example, when a user requests directions, Google maps returns a small set of routes, and the user typically picks her final route from those routes (Figure \ref{fig:noisy-example}). }. A human is \emph{noisy} if she has inherent randomness in her ability to recognize her true preferences over items. ``\emph{To err is human}'', and it has been widely recognized that humans are often imperfect at picking the \enquote{correct} item (e.g., see \citep{agranov2017stochastic, plackett1975analysis, luce1959individual, hey1994investigating}). A human is \emph{misaligned} with another agent (human or algorithm) if they have different ground-truth preferences over items, which is distinct from their potentially noisy \emph{realizations} of those preferences. For example, if Alice prefers horror movies and Bob prefers comedies, then their ground truth preferences over a set of movies would be \emph{misaligned} with each other. Additionally, given a single algorithmic recommender, either Alice or Bob (or both) would be \emph{misaligned} with the algorithm. Of course, in the limiting case of perfect personalization, both could have separate algorithms that are perfectly aligned with their preferences, but in realistic settings where one algorithm must serve a diverse range of humans, misalignment is a reality of life. Note that our framework has strong connections to other bodies of work, such as pluralistic alignment and conformal prediction: we discuss such connections in Section \ref{sec:related}.
 
One of the key points of our paper is that there are settings where misalignment can be \emph{helpful}. As a stylized example, consider \Cref{fig:noisy-example}, which illustrates a collaboration between Google Maps (algorithm) and Alice (human). Google Maps recommends $k=2$ routes $\{x_1, x_3\}$ to Alice from three routes $\{x_1, x_2, x_3\}$.  The true preference of Alice is $(x_1, x_2, x_3)$, but she cannot always identify her preference exactly (potentially because she only has access to imperfect information about the outcomes) and has a probability of $0.1$ of wrongly ranking them as $(x_2, x_1, x_3)$. As a result, she has a probability of $0.9$ to choose her favorite route $x_1$ when making the decision alone. However, when following the suggestions of Google Maps, she is always able to rank $x_1$ before $x_3$ and choose the best route for her. The high-level goal of our paper is to study when settings like those in \Cref{fig:noisy-example} occur and when a human benefits from relying on a potentially misaligned algorithm. 
\begin{figure}[t]
\centering
\includegraphics[width=0.8\linewidth]{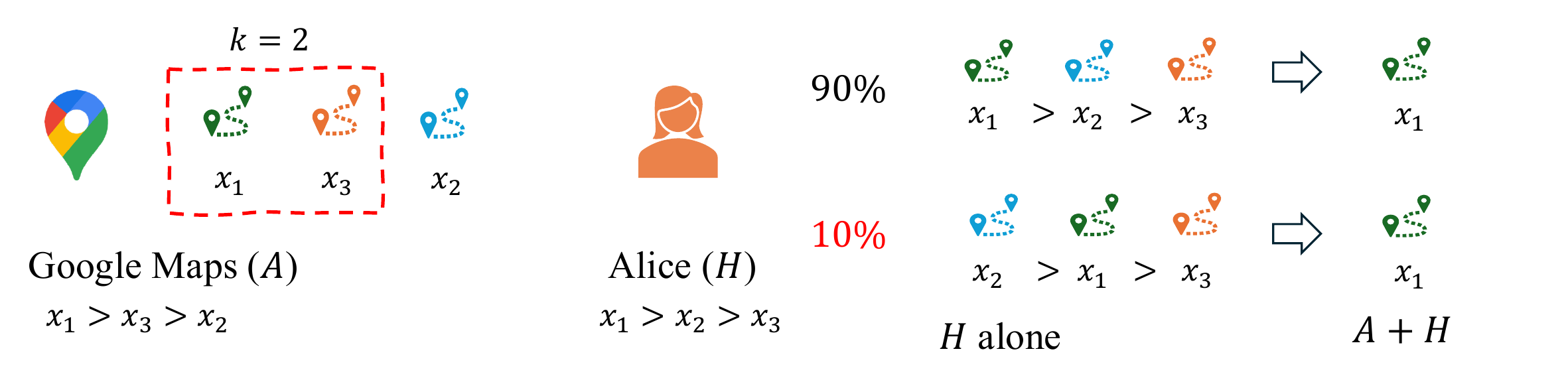}
\caption{Human-AI collaboration between Google Maps (Algorithm) and Alice (Human). Here, we assume that the algorithm is deterministic, but the human is noisy and only the best item ($x_1$) has value. If Alice picked by herself, she would pick $x_1$ 90\% of the time, but when the algorithm deterministically reduces her set to $\{x_1, x_3\}$, she picks $x_1$ always. }
\label{fig:noisy-example}
\end{figure}

\paragraph{Our contributions.} 
In this work, we begin by formally defining our model in Section \ref{sec:model}: this includes a description of how the human and algorithm interact, as well as the types of noise models that the human and/or algorithm exhibit, and gives a formal definition of misalignment. Additionally, we define several key objectives for the human-algorithm system to satisfy. For example, one objective is maximizing social welfare: if we assume each person has some utility for each item, expected (utilitarian) social welfare gives the total expected utility of all humans using the same algorithmic curation tool. Other objectives relate to the comparative performance of the human-algorithm system and the human by herself: does adding the algorithm increase the human's expected welfare, or decrease it? We say that a system achieves \emph{uplift} if it increases the utility of every human, relative to what they would have achieved by themselves.
Throughout the rest of this work, our goal will be to study when each of these objectives can be satisfied. We assume that an algorithm can be modified by either changing its ground-truth ordering over items (i.e., changing the training objective), or changing the noise level of the algorithm (i.e., changing the temperature in a stochastic model).

Our first problem focuses on when misalignment is helpful for an individual human: 
\begin{problem}[Comparative Benefits of Misalignment]
Given a fixed noisy human, which algorithm leads to the greatest utility for the human? 
\end{problem}

In Section \ref{sec:whobenefitsmore}, we begin by studying a fixed human who is considering a range of algorithmic tools that are all misaligned to varying degrees with the human. First, we note that it is possible for a human to benefit when using a misaligned algorithm: this should be intuitive to anyone who has successfully gained value out of an algorithmic tool that is imperfectly calibrated to their own preferences. Next, this section provides a comparative analysis of which type of algorithm assists the human more or less.  Specifically, we show that increased alignment with the algorithm does \emph{not} guarantee increased performance, and formally describe when misalignment may be helpful for utility. To give intuition, misalignment leads to \enquote{\emph{different kinds of mistakes}} which can sometimes be helpful, if the mistakes are not too costly and if it makes it easier for the human to correctly identify their preferred items.

Next, Section \ref{sec:strathum} builds on this framework to answer the question, \enquote{How should we design an algorithm system to benefit a \emph{population} of human users}?

\begin{problem}
Can we find a ground-truth ranking and accuracy parameter such that the joint system maximizes utilitarian social welfare? Can we find settings such that we achieve uplift, or is such a setting impossible? 
\end{problem}
 Our main contributions can be summarized (informally) as follows,
\begin{itemize}[leftmargin=0.5cm] 
   \item An algorithm that maximizes utilitarian welfare is always noiseless. While it is in general computationally hard to find this algorithm, we can design a mixed integer program (MIP) which is similar to the MIP used for assortment optimization under the Mallows model~\citep{desir2016assortment}.
   \item However, we show cases where uplift can only be achieved by a noisy algorithm, aligning with the broader literature on the importance of randomization for fairness goals and providing motivation for the default non-zero temperature settings that many generative tools use. 
  \item We provide efficient algorithms to validate whether a given algorithm achieves uplift, and identify natural sufficient conditions for when we can design algorithms that always achieve uplift.
\end{itemize} 

While most of our contributions are theoretical, in each section we explore generalizations of our results through numerical simulations, showing that our main contributions generalize substantially beyond our formal model. Finally, in \Cref{sec:discuss} and \Cref{sec:extensions}, we conclude by summarizing our main contributions, discussing extensions, and highlighting impacts for policymakers, platform designers, and applied researchers. At a high level, we view our work as highlighting an important tension between the desiderata of creating algorithms that \emph{think like us} and \emph{provide additional value} over what humans could do by themselves. 

\section{Model}\label{sec:model}
\subsection{Human-algorithm Collboration Model}
In a human-algorithm collaboration model, there are two actors: an algorithm ($A$) and a human ($H$).
There is a set of items $M = \{x_1, \ldots, x_m\}$ representing different outcomes (e.g., labels in prediction tasks, generative model output). The setting we study is that of where the algorithm acts as a \emph{curator}, picking a subset of size $k$ from a larger set of items of size $m$. 
The specific mode of interaction we assume is where the algorithm picks their top $k$ from some noisy ranking over items, and the human picks their favorite among that set, according to their own noisy ranking over items\footnote{Note that this is the same modeling assumption as \citep{donahue2024listsbetteronebenefits} in the unanchored case, but our results significantly generalize over their results by allowing agents to be misaligned.} This setting may occur in settings such as content recommendation (e.g., Netflix or Spotify) or algorithmic assistants in predictive tasks. 

These items can have different values to the actors.
One notable special setting is called \textbf{top item recovery}, where only a single item has positive value. For example, these settings where the goal of the human is to recover the single best item, such as a job candidate who is most suitable, or the \enquote{best} route between two locations. 

We say the algorithm and the human are \emph{aligned} if they agree on the values of each item. However, as two actors may get information from different sources, perfect alignment cannot always be achieved. We consider \emph{misalignment} within the actors through the following two aspects: either between the algorithm and the human, or between different humans. Formally, the humans come from a population of $n$ types, where the $i$th type has probability $p_i$ to occur. If human and algorithm are misaligned, but their top items are the same, we say they are \emph{top-aligned}; otherwise, they are \emph{top-misaligned}.
We assume that human has descending values according to their ground-truth rankings.
Let $v_{i,j}$ denote the value of the $j$-th best item for a human of type $i$.
For each human of type $i \in [n]$, let her utility function be $u_i$, where $u_i(x) = v_{i,j}$ if item $x$ is the $j$-th item in her ground-truth ranking $\hr_i^*$.
We denote by $x_H^i$ the item chosen by the human of type $i$ when acting alone, by $x_A$ the item chosen by the algorithm when acting alone, and by $x_C^i$ the item chosen by the human of type $i$ under collaboration.

\subsection{Noisy Ranking Distribution}
We assume that each agent can only access a noisy permutation over items: this may be because the agent only has access to imperfect or noisy data about each alternative, for example, or due to features such as temperature in generative AI tools~\citep{ackley1985learning}. 
Denote the sampled rankings of the algorithm and human as $\ar \sim \D_a$ and $\hr_i \sim \D_h^i$ respectively, where $\D_a$ concentrates at ground-truth ranking $\ar^*$ and $\D_h^i$ concentrates at $\hr_i^*$. 
Denote by $x_i \ogt{\pi} x_j$ that $x_i$ precedes $x_j$ in $\pi$.
Denote by $\ari{i}$ and $\ars{k}$ the $i$-th item and the first $k$ items (viewed as an unordered set) in a ranking $\ar$.  
Denote by $\pi\circ (x_i x_j)$ the ranking after swapping the locations of $x_i$ and $x_j$ in ranking $\pi$.
All distributions in the paper are \emph{inversion-monotonic} and \emph{label-invariant}. 
Informally, adding inversion to a ranking only decreases its probability, and the distribution is not tied to the labels of the items.
Both Mallows model~\citep{mallows1957non} and Plackett-Luce model~\citep{luce1959individual,plackett1975analysis} satisfy the two properties.

\begin{definition}[Inversion-monotonicity and label-invariance]
A distribution $\sigma \sim \D(\sigma^*)$ with ground-truth ranking $\sigma^*$ is \emph{inversion-monotonic} if $\myP[\sigma = \sigma_1] > \myP[\sigma = \sigma\circ (x_i, x_j)]$ where $x_i \succ_{\sigma^*} x_j$ and $x_i \succ_{\sigma_1} x_j$ and \emph{label-invariant} if relabeling the items (i.e., applying an arbitrary permutation to all items) does not change the probability of any ranking.
We say that two distributions are \emph{isomorphic} if one is identical to the other by relabeling the items.
\end{definition}

\paragraph{Mallows Model.}
\label{app:def_of_mallows}
The Mallows model~\citep{mallows1957non} generates a distribution over permutations based on their number of inversions relative to a central ranking (also known as the Kendall-Tau distance).
A Mallows model $\D(\pi^*, \phi)$ consists of a central ranking $\pi^*$ and an accuracy parameter $\phi \ge 0$. The probability of a permutation $\pi$ occurring is $\myP[\pi]= \frac1{Z}\cdot \exp(-\phi\cdot d(\pi^*, \pi))$, where $Z$ is the normalization constant and $d(\pi^*, \pi)$ is the Kendall-Tau distance between $\pi$ and $\pi^*$.
We assume all accuracy parameters used in this paper are \emph{positive}.

\paragraph{Plackett-Luce Model.}
\label{app:def_of_plackett_luce}
The random utility model (RUM) is commonly used as a model of permutations~\citep{thurstone1927law}, where a permutation is generated as follows: i.i.d. noise is added to the true value of each item to produce $\hat{v_i} = v_i + \epsilon$, for $\epsilon \sim \mathcal{D}$: then, the items are sorted in order of the noised values $\{\hat{v_i}\}$. 
Note that in a RUM, the probability of sampling a particular ranking depends not only on the number of inversions but also on the specific valuations of the items. 
Although there is no definition of  ``central ranking'' in RUM, we refer to the ranking consistent with the item values as the central ranking.
When $\epsilon$ follows i.i.d. Gumbel noise $G(\mu, \beta)$, following the analysis in \citep{luce1959individual,plackett1975analysis}, the probability of a permutation $\sigma$ is given by
\begin{align}
\myP[\pi] = \prod_{j=1}^m \frac{\exp(v_{\pi_j}/\beta)}{\sum_{\ell=j}^m \exp(v_{\pi_\ell}/\beta)},
\end{align}
where $v_{\pi_i}$ is the value of the $i$-th item in the ranking $\pi$. 
In particular, as $\beta$ increases and item values become more heterogeneous, the probability of sampling a ground-truth ranking also increases.

\subsection{Welfare objectives}
We are interested in how an algorithm performs over a \emph{population} of people who are all using the algorithm simultaneously, that is, \emph{welfare objectives} over different humans who may receive different utility from the same algorithm. There are two main welfare objectives we will consider: utilitarian social welfare and uplift. Utilitarian social welfare is equal to the sum of the human's expected utilities (i.e., $\myE[u_i(x_C^i)]$) weighted by the probabilities of every type of human (i.e., $p_i$).
Formally,
\begin{restatable}[Expected utilitarian social welfare]{definition}{socialwelfare}\label{def:socialwelfare}
The \emph{utilitarian social welfare} of a joint system is the sum of expected utility of the $n$ types of humans $\sum_{i=1}^n p_i\cdot \myE[u_i(x_C^i)]$.
\end{restatable}

Further, we are interested in whether the human-AI collaboration can lead to improvements in human utility, compared to the utility the human would experience if she solved the task herself. If an algorithm can simultaneously benefit a population of people in this way, we say that the algorithm achieves \emph{uplift}. 

\begin{restatable}[Uplift]{definition}{compuplift}\label{def:compuplift}
A human-algorithm joint system achieves \emph{uplift} if the expected utility of every type of human is improved: 
$\myE[u_i(x_C^i)] >  \myE[u_i(x_H^i)]$ for any $i\in [n]$\footnote{We refer the reader to the related work section for discussion about the relation between uplift and complementarity \citep{bansal2021does}}.
\end{restatable}

\section{Related Work}
\label{sec:related}

\paragraph{Human-algorithm collaboration.}
Our work relates to the general area of human-algorithm collaboration. In particular, there is a rich history of applied and empirical work in human-algorithm interaction: we refer interested readers to see \citep{preece1994human, kim2015human, mackenzie2024human, lazar2017research} for textbook treatments. Specifically, our work relates to a growing literature using theoretical models to analyze humans interacting with algorithms.
Some works study how to design algorithms to optimally assist humans \citep{bansal2021most, chan2019assistive, donahue2022human, madras2017predict, Agarwal2022DiversifiedRF}, including work incorporating models of human cognitive biases, such as \citep{li2024decoding, ibrahim2025measuring, chen2025missing}. Other work decides when human-algorithm teams perform well \citep{greenwood2024designing, peng2024no, steyvers2022bayesian, agarwal2023online, cowgill2020algorithmic, benCSCW, alur2023auditing, alur2024human, guo2025value}, often relating to benchmarks such as complementarity (strict improvement over the human or algorithm alone defined in \citep{bansal2021does}). Some relevant literature reviews, taxonomies, and systematic studies in this space include \citep{gomez2025human, vaccaro2024combinations, rastogi2022unifying}.

Within human-algorithm collaboration, our work is most closely related to that of~\cite {donahue2024listsbetteronebenefits}, which studies a similar setting where an algorithmic tool presents a top $k$ subset to a noisy human. A key difference in our work is that we allow humans to be \emph{misaligned} with the algorithm and study when misalignment is helpful. In fact, many of our contributions strictly generalize theirs, such as generalizing the utility function beyond top item recovery. Other related areas include \emph{conformal prediction}, which studies how to optimize a subset of items (e.g., ensuring that the best item is presented with high probability): see \citep{angelopoulos2023conformal, fontana2023conformal} for a summary of work in this area. Within this space, some works focus on optimizing the set of items that are presented \citep{straitouri2022provably, wang2022improving, straitouri2023designing, de2024towards, babbar2022utility, de2024towards, hullman2025conformal, zhang2024evaluating}, while others include more empirical analyses of specific settings \citep{angelopoulos2020uncertainty, arnaiz2025towards}. In general, these works do not consider settings with multiple humans who may be misaligned with each other: one exception is \citep{corvelo2025human}, which studies an empirical setting on when algorithmic alignment is a helpful property for tools assisting human decision-makers. 
There is also a line of work studying human-AI collaboration through information elicitation~\citep {steyvers2022bayesian,corvelo2023human,corvelo2025human,collina2024tractable,collina2025collaborative}.
\cite{steyvers2022bayesian} developed a Bayesian framework to combine the predictions of humans and AI algorithms with confidence scores. 
\cite{corvelo2023human,corvelo2025human} then showed that if the confidence value aligns with the human's confidence, the collaboration will benefit the human's decision-making.
\cite{collina2024tractable,collina2025collaborative} studied the setting where each party holds different feature information and transmits their numerical prediction to each other at each round.
\cite{collina2024tractable} generalizes Aumann’s Agreement Theorem by relaxing the rationality assumptions and introducing calibration-based conditions on each party that ensure efficient convergence of the conversation.
Building on this, \cite{collina2025collaborative} then proposes more communication-efficient protocols and removes the assumption of a common and correct prior. Finally, \cite{collina2025emergent} models a Bayesian communication setting where a human is relying on information from multiple AI tools that are misaligned with each other (and with her) to make a discrete decision. Compared to their work, our work does not model the interaction of the human and the algorithm as a multi-round process.
The algorithm acts more like a curator, instead of a rational agent. 
\paragraph{Pluralistic alignment.}
Another research area that our paper relates to is pluralistic alignment. In this setting, the goal is generally to \emph{align} an algorithm tool with users who have heterogenous preferences (e.g. \citep{sorensen2024roadmap}), especially the work connected with voting and social choice literature on aggregating diverse user preferences (e.g. \citep{conitzer2024social, dai2024mapping, shirali2025direct, a2024policy, ge2024axioms, pardeshi2024learning, chen2024pal, siththaranjan2023distributional, golz2025distortion}). One key difference between much of this work and our own is the meaning of noise. Often, work in this space (and voting in general) assumes that users know their own preferences and can represent them faithfully, which in the case of political opinions is often a fairly reasonable assumption. In our setting, we are assuming that humans may be able to only noisily access their true values over items and thus can benefit from the assistance of another agent (e.g., an algorithm). This is a more natural assumption in lower-stakes settings such as content recommendation or many types of content creation, such as writing code or daily emails. Separately, some papers study the performance of voting rules under noise (e.g., sporting competition), though they tend to be less focused on societal welfare objectives ~\citep{boehmer2022quantitative, 10.1145/2892565}. 

\paragraph{Complementarity and Uplift.} 
Note that the term of uplift is related to that of \emph{complementarity} as defined in \citep{bansal2021does}: ``\emph{a human-AI system achieves complementary performance if it outperforms both the AI and the human acting alone}''. As compared to complementarity, uplift is weaker because it only requires that the human-algorithm system outperform the human alone, rather than the human and algorithm jointly. However, it also differs (and is weakly stronger) in that it is a \emph{population-level} property: an algorithm must benefit all humans using it in order to achieve uplift.
In addition, the application scenarios of complementarity and uplift also differ. 
The classical setting of complementarity usually assumes that both the human and the algorithm aim to predict an objective event.
In contrast, our model does not assume the existence of an objective ranking of the items, and the goal of human-algorithm collaboration is to better align humans' preferences.
Therefore, the performance of the algorithm alone is not well-defined in our setting and depends on how well the algorithm meets the human's preferences.

\paragraph{Other ranking/permutation models.} The basic ranking model considered in our paper is defined over the number of inversions.
A wide range of alternative permutation models has been studied in the literature.
For example, the Bradley-Terry model~\citep{bradley1952rank} and the Plackett-Luce model~\citep{plackett1975analysis,luce1959individual} view a permutation as the consequence of a sequence of choices; the weighted Mallows model~\citep{raman2014methods} assigns weights to the inversions and defines the probability of a permutation by weighted Kendall-Tau distance.
Moreover, \cite{awasthi2014learning,liu2018efficiently} studied learning a mixture of the Mallows model, which assumes the sampled permutation comes from a heterogeneous population.
\section{Single Human with AI Assistant: Benefits to Misalignment}\label{sec:whobenefitsmore}
In this section, we begin by analyzing a single human who is considering multiple different algorithmic assistants. These algorithmic assistants differ from each other in their relative ground-truth orderings over items - e.g., some of them may be \emph{misaligned} with the human's ordering over items. At first glance, we may expect the human's utility to decrease with increased misalignment with the algorithm. However, as mentioned by the previous example in \Cref{fig:noisy-example}, misalignment sometimes better assists a human in making the decision. The goal of this section is to describe what those conditions are. 

Without loss of generality, we relabel items by the human's ranking: $\hr^* = (x_1, \ldots, x_m)$. 
Meanwhile, we omit the index $i$ for the human type and denote the human’s value for the $j$-th item by $v_j$.
All the proofs in this section can be found at \Cref{app:omitted_proofs_of_sec4}.

\subsection{Top-item Recovery and Related Setting}
As a warm-up, we first show a fine-grained analysis of the potential effect on the human's decision-making after swapping two specific items in the algorithm's ground-truth ranking.

Let $x_i, x_j$ be two items with $i<j$, indicating that item $x_i$ is (weakly) better than item $x_j$ to the human.
Denote algorithms $A_1$ and $A_2$ that are identical except that $A_1$ places $x_i$ before $x_j$ and $A_2$ inverts a pair of items $x_i$ and $x_j$ in their ground-truth rankings.
We abuse the notations $x_C^1$ and $x_C^2$ for the items that the human picks when collaborating with $A_1$ (the aligned one) and $A_2$ (the misaligned one), respectively.
We observed that, 

\begin{restatable}{lemma}{LemProbOfiP}
\label{lem:prob_of_picking_item_i1}
For any item that is not $i$, the probability of picking that item is higher with the misaligned algorithm (and the probability of picking item $i$ is lower). That is, for any $r\in [m]$ with $r\neq i$, $\myP[x_C^1 = x_r] \le  \myP[x_C^2 = x_r]$, and $\myP[x_C^1 = x_i] \ge \myP[x_C^2 = x_i]$\footnote{Note that the two inequalities become tight only when the two items are completely indistinguishable to the human, i.e., swapping them does not change the probability of any ranking under the human’s distribution. The following content considers the case where this condition does not hold.
}.
\end{restatable}

\paragraph{Proof sketch.} We highlight the ideas here, and the full proof can be found at \Cref{app:omitted_proofs_of_sec4}.
Since the misaligned algorithm $A_2$ places $x_i$ lower than $x_j$ in the ground-truth ranking compared to $A_1$, it is more likely to present $x_r, x_j$ rather than $x_r, x_i$ to the human.
As $x_j$ is relatively worse than $x_i$ in the human’s preference, the human is thus more likely to choose $x_r$ when presented with $x_r, x_j$.

The itemwise probability comparison then implies \Cref{thm:benefit_inv1} below, which distinguishes between misalignment on items that are valueless to the human and misalignment with items that are maximally valuable to the human, showing that the former are always helpful and the latter are always harmful. 

\begin{restatable}{theorem}{ThmBenefit}
\label{thm:benefit_inv1} 
The human gains \textbf{higher} utility with the misaligned algorithm when $x_i$ and $x_j$ are least valued, but \textbf{lower} utility when $x_i$ is top valued.
\end{restatable}

We illustrate \Cref{thm:benefit_inv1} in Example \ref{ex:illustrate_benefit_inv1}  below. 

\begin{example}\label{ex:illustrate_benefit_inv1}
Consider the setting with three items, where the human’s ground-truth ranking is $\hr^* = (x_1, x_2, x_3)$, and there are four potential algorithms (\Cref{fig:illustration_of_who_benefit_more1}).
Suppose items $x_2$ and $x_3$ have equal value to the human.
By repeated applications of Theorem \ref{thm:benefit_inv1}, we can derive relationships between the value of each of these algorithms to the human: for example, Algorithm 2 is better than Algorithm 1 because it is created by an inversion in value-less items ($x_2, x_3$), and Algorithm 3 is better than Algorithm 4 by identical reasoning. Similarly, Algorithm 1 is better than Algorithm 4 because Algorithm 4 involves an inversion with the most valuable item (and Algorithm 2 is better than Algorithm 3 by identical reasoning). 
\begin{figure}[h]
\centering
\begin{tikzpicture}[
    blueroundnode/.style={circle, draw=blue!70, fill=blue!5, very thick, minimum size=3mm},    
    roundnode/.style={circle, draw=black, thick, minimum size=3mm},
]
\tikzset{
    state/.style={circle, draw, minimum size=3mm},
    dashedline/.style={dashed, rounded corners}
}

\node[] at (-4, 0) {$\hr^*$};
\node[blueroundnode, scale=0.8] (a1) at (-5, -1) {$1$};
\node[roundnode, scale=0.8] (a2) at (-4, -1) {$2$};
\node[roundnode, scale=0.8] (a3) at (-3, -1) {$3$};

\node[] at (0, 0) {$A_1$};
\node[blueroundnode, scale=0.8] (x1) at (1, 0) {$1$};
\node[roundnode, scale=0.8] (x2) at (2, 0) {$2$};
\node[roundnode, scale=0.8] (x3) at (3, 0) {$3$};

\node[] at (8, 0) {$A_2$};
\node[] at (4, 0) {\Large $<$};
\node[blueroundnode, scale=0.8] (y1) at (5, 0) {$1$};
\node[roundnode, scale=0.8] (y2) at (6, 0) {$3$};
\node[roundnode, scale=0.8] (y3) at (7, 0) {$2$};

\node[] at (2, -0.6) {\large $\vee$};
\node[] at (0, -1.2) {$A_4$};
\node[roundnode, scale=0.8] (z1) at (1, -1.2) {$3$};
\node[roundnode, scale=0.8] (z2) at (2, -1.2) {$2$};
\node[blueroundnode, scale=0.8] (z3) at (3, -1.2) {$1$};

\node[] at (4, -1.]) {\large $<$};
\node[] at (8, -1.2) {$A_3$};
\node[] at (6, -0.6) {\large $\vee$};
\node[roundnode, scale=0.8] (w1) at (5, -1.2) {$3$};
\node[blueroundnode, scale=0.8] (w2) at (6, -1.2) {$1$};
\node[roundnode, scale=0.8] (w3) at (7, -1.2) {$2$};
\end{tikzpicture}
\caption{Illustration of \Cref{thm:benefit_inv1} with $3$ items, where the rounded node with number $i$ represents item $x_i$ and the most valuable item is in blue.}
\label{fig:illustration_of_who_benefit_more1}
\end{figure}
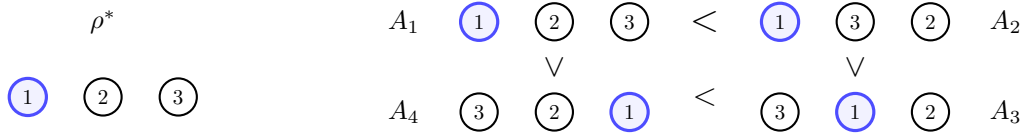
\end{example}

As a corollary of \Cref{thm:benefit_inv1}, \Cref{thm:topitem} characterizes the best and the worst algorithm assistants for the top item recovery setting (where each human only has value for recovering her favorite item).

\begin{restatable}[Best/worst strategy for top item recovery]{theorem}{CoroTopItem}
\label{thm:topitem}
In the top item recovery setting, the algorithm's ground-truth ranking $\ar^*$ that maximizes human's expected utility is $\ar^*= (x_1, x_m, \ldots, x_2)$ while the one minimizing human's expected utility is $\ar^*= (x_2, \ldots, x_m, x_1)$.
\end{restatable}

These results imply that \Cref{thm:benefit_inv1} gives a \emph{partial order} over which type of algorithms a given human would prefer. One natural question would be whether a \emph{total order} over algorithms is possible.
The following example illustrates why this is not always possible.
\begin{example}
Consider the human working with two algorithm assistants $A_1$ and $A_2$, of which the ground-truth rankings are respectively $(x_1, \dots,x_{10})$ and $(x_1,\dots,x_5, x_{10}, \dots, x_6)$.    
Note that the design of $A_2$ follows the insight of \cref{thm:benefit_inv1} - aligning with the human on the high-valued items, while reversing the low-valued items.
\Cref{fig:comparison_a_ta} plots the relative performance between the two algorithmic assistants ($A_1$ and $A_2$) and the human working alone ($H$), where all of them follow the Mallows distribution, and the x-axis/y-axis are respectively the accuracy of the human and the algorithm.
We assume that the two algorithmic assistants have the same accuracy. 
\begin{figure}[h]
\centering
\includegraphics[width=0.4\linewidth]{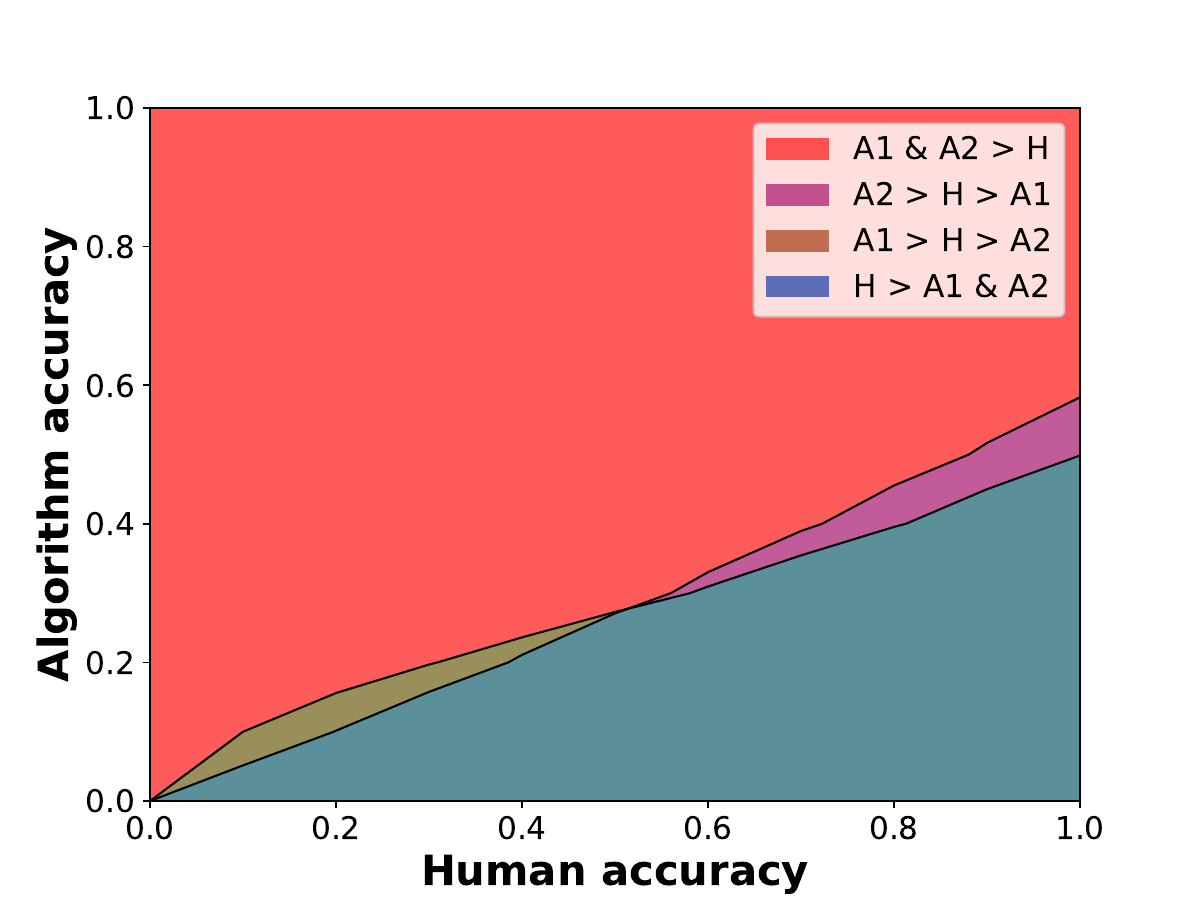}
\caption{Comparison between human working with $A_1$, $A_2$, or alone, where $A_1$ \& $A_2 > H$ denotes that the human performs better when assisted by either algorithm.}
\label{fig:comparison_a_ta}
\end{figure}
It can be noticed that the relative performance between $A_1$ and $A_2$ is not fixed and depends on the accuracy level.
\end{example}

\subsection{Extensions to Specific Distributions}
\begin{table*}[t]
    \centering
    \small
    \begin{tabular}{c||c|c|c|c}
    \toprule 
    & $x_i$ (top), $x_j$ ($<$ top) & $x_i, x_j$ (nearly top) & $j - i >\Delta$, $x_i, x_j$ ($<$ top) & $x_i, x_j$ (bottom) \\
    \hline
    Highly accurate human & \cellcolor{blue!20} $A_1 > A_2$ (\cref{thm:benefit_inv1}) & \cellcolor{blue!20} $A_1 > A_2$ (\cref{lem:less_utility_k_2_mallows}) & \cellcolor{green!20} $A_2 > A_1$ (\cref{lem:more_utility_k_2_mallows}) & \cellcolor{green!20} $A_2 > A_1$ (\cref{thm:benefit_inv1}) \\
    \hline
    Accurate human  & \cellcolor{blue!20} $A_1 > A_2$ (\cref{thm:benefit_inv1}) & \cellcolor{blue!20} $A_1 > A_2$ (\cref{lem:less_utility_k_2_mallows}) & \cellcolor{green!20} $A_2 > A_1$ (\cref{lem:more_utility_k_2_mallows})  & \cellcolor{green!20} $A_2 > A_1$ (\cref{thm:benefit_inv1})  \\
    \hline 
    Highly noisy human  &  \cellcolor{blue!20} $A_1 > A_2$ (\cref{thm:benefit_inv1}) & \cellcolor{blue!20} $A_1 > A_2$ (\cref{lem:less_utility_k_2_mallows}) & \cellcolor{blue!20} $A_1 > A_2$ (\cref{lem:less_utility_k_2_mallows}) & \cellcolor{green!20} $A_2 > A_1$ (\cref{thm:benefit_inv1}) \\
     \hline
    \end{tabular}
    \caption{Summary of the insights by \Cref{thm:benefit_inv1} and its extensions. Note that $A_1$ is aligned and $A_2$ is misaligned, creased by swapping items $x_i, x_j$.
    $x_i$ (top or bottom)  indicates that $x_i$ is the most (or least) valuable item to the human.
    $j - i > \Delta$ indicates that there is an index gap between item $x_i$ and $x_j$. 
    This table summarizes when misalignment is helpful or harmful. 
    }
    \label{tab:placeholder}
\end{table*}

Another question is how \Cref{thm:benefit_inv1} generalizes to settings where the inverted items \( x_i \) and \( x_j \) may not necessarily be the most or least valuable items to the human.
This section investigates this question under two specific models: the Plackett-Luce model and the Mallows model by quantifying the probability changes in \Cref{lem:prob_of_picking_item_i1}. 

We extend \Cref{thm:benefit_inv1} both theoretically and numerically: 
\textbf{(i)} We establish more general conditions under which misalignment is guaranteed to be either beneficial or harmful under the Mallows and Plackett–Luce models.
These theoretical extensions provide a characterization of when swapping $x_i$ and $x_j$ is beneficial or not.
We summarize the insights in \Cref{tab:placeholder}.
\textbf{(ii)} We present a numerical study in \Cref{app:extension_of_theorem1}, examining scenarios where both the human and the algorithm follow either a Mallows model or a Random Utility Model (RUM). 
Our results indicate that the conclusions of \Cref{thm:benefit_inv1} continue to hold even when \(x_i\) and \(x_j\) take relatively small values.

Let $\pi_1$ and $\pi_2$ respectively be the rankings output by $A_1$ and $A_2$.
The following lemma first provides a sufficient condition when misalignment is harmful to the human.

\begin{restatable}{lemma}{lemLessUtilityMallows}
\label{lem:less_utility_k_2_mallows}
Sufficient condition for when collaboration with a misaligned algorithm is \textbf{harmful} for the human under the Mallows model with accuracy parameter $\phi_h$ is $v_1 - v_i\le \frac{\exp(-\phi_h \Delta)}{1 - \exp(-\phi_h\Delta)}\Delta$ where $\Delta = j - i$ and $k=2$.
\end{restatable}
We now provide further intuition on when the above conditions hold. 
The condition is more likely to hold when the value of the most valuable item, $v_1$, does not significantly exceed the value of $x_i$, $v_i$, or $\phi_h$ is sufficiently small.
The former condition means that $x_i$ is nearly the top item, while the latter one indicates the human is \emph{highly noisy}.
Under these conditions, collaboration with the misaligned algorithm tends to only hurt the human.

It is worth noting that \Cref{lem:less_utility_k_2_mallows} does not provide a necessary condition: when $\Delta$ is large, the right-hand side tends to zero, and the condition is hard to satisfy.
However, the human may still be harmed in this case.
For instance, \Cref{thm:benefit_inv1} shows that the human always receives less utility once $x_i$ is the most valuable item while $x_j$ is less valuable.

\begin{takeaway}
The human gets hurt in working with the misaligned algorithm if either (i) the human is \emph{highly inaccurate}, or (ii) both $x_i$ and $x_j$ are \emph{nearly} top items.
\end{takeaway}

Next, we complement \Cref{lem:less_utility_k_2_mallows} by providing sufficient conditions when misalignment is beneficial to the human. 

\begin{restatable}{lemma}{lemMoreUtilityMallows}
\label{lem:more_utility_k_2_mallows}
Sufficient condition for when collaboration with a misaligned algorithm is \textbf{helpful} for the human under the Mallows model is when some $i' \in [i-1]$ satisfies
\begin{align*}
\frac{v_{i'}}{v_i}\ge \frac{\sum_{r \neq i, j} \psi(i, j, r)}{\sum_{r=1}^{i'} \psi(i, j, r)\exp(-\phi_h\cdot (j- r +1))}  \frac{1}{1-\exp(-\phi_h\cdot \Delta)},
\end{align*}
where $\Delta = j-i+1$ and $\psi(i, j, r) = \myP[\pi_1[:2] = \{x_i, x_r\}] - \myP[\pi_1[:2] = \{x_j, x_r\}]$ is always nonnegative.
\end{restatable}
Notice that $\psi(i, j, r)$  on the right-hand side depends only on the algorithmic assistant’s distribution.
In the case where most of the values of $\psi(i,j,r)$ concentrate on terms with $r \le i'$ (e.g., the algorithm is aligned with the human on the high-valued items), the first term asks for an exponential gap between $v_{i'}$ and $v_i$.
Meanwhile, the second term depends on the index gap between the two swapped items (i.e., $\Delta$) and the accuracy level of the human (i.e., $\phi_h$).
The condition is more likely to hold when the index gap and the accuracy level of the human are both non-trivial, i.e., larger than some non-negligible constants.
The following provides a concrete example where the condition holds.
\begin{example}
Consider the setting of four items and two algorithm assistants that randomly select two items from the top three of their ground-truth rankings, displaying the top two items with a higher probability than displaying other subsets.
Suppose the ground-truth rankings of $A_1$ and $A_2$ are $x_1, x_2, x_3, x_4$ and $x_1, x_3, x_2, x_4$ respectively, where $x_2$ and $x_3$ are the misaligned items. 
The human has value of $v_1= 100, v_2 = 2,  v_3 = 1, v_4 = 1$.
The human's ranking satisfies the Mallows model with accuracy parameter $\phi_h = 1$.
Let $i'=1$. Then the two sides of \Cref{lem:more_utility_k_2_mallows} are respectively given by
\begin{align*}
LHS = \frac{v_1}{v_2} = \frac{100}{2} = 50\quad 
RHS \le \exp(3)\cdot \frac{1}{1 - \exp(-1)} < LHS\,.
\end{align*}
Therefore, the human receives a higher utility when collaborating with the misaligned algorithm $A_2$ than the aligned one $A_1$.
\end{example}

\begin{takeaway} The human benefits from collaborating with a misaligned algorithm when there exists an item that is \emph{significantly more valuable} than $x_i$ and $x_j$, and (i) the human possesses a \emph{non-trivial} level of accuracy, and (ii) there is a clear index gap between $x_i$ and $x_j$ in their ranking.
\end{takeaway}

Lastly, we remark that a similar set of extensions also holds under the Plackett–Luce model.
When the human's distribution satisfies the Plackett-Luce model with noise parameter $\beta$, we have the following two results.

\begin{restatable}{lemma}{lemLessUtility}
\label{lem:less_utility_k_2_plackett_luce}
Sufficient condition for when collaboration with a misaligned algorithm is \textbf{harmful} for the human under the
Plackett-Luce model is when $v_0 - v_j \le 1.27\beta$.
\end{restatable}

\begin{restatable}{lemma}{lemMoreUtility}
\label{lem:more_utility_k_2_plackett_luce}
Sufficient condition for when collaboration with a misaligned algorithm is \textbf{helpful} under the Plackett-Luce model for the human is when some $i' \in [i-1]$ satisfies
\begin{align*}
\frac{v_{i'}}{v_i} \ge \frac{\sum_{r\neq i, j} \frac{\psi(i, j, r)}{\exp((v_r - v_i)/\beta) + 1}}{\sum_{r=1}^{i'} \frac{\psi(i, j, r)}{\exp((v_r - v_i)/\beta) + 1}} \cdot \frac{2}{1- \exp\left(-\Delta/\beta\right)},
\end{align*}
where $\Delta  = v_i - v_j$ and $\psi(i, j, r) = \myP[\pi_1[:2] = \{x_i, x_r\}] - \myP[\pi_1[:2] = \{x_j, x_r\}] \ge 0$.
\end{restatable}

At a high level, similar to \Cref{lem:less_utility_k_2_mallows}, when the human is highly noisy or $x_i$ and $x_j$ are nearly top items, misalignment is harmful.
Slightly different from \Cref{lem:more_utility_k_2_mallows}, the positive result under the Plackett–Luce model considers the value gap between the two items (i.e., $v_i - v_j$), which is consistent with the definition of the distribution.

\section{Multiple Humans Using Single AI: Welfare Objectives}\label{sec:strathum}

In Section \ref{sec:whobenefitsmore}, we studied the question of which (potentially misaligned) algorithm an \emph{individual} human would prefer.
In this section, we study the question of which algorithm to select in order to satisfy welfare notions over a \emph{population} of humans. 
In particular, this section focuses on the Mallows model (the rankings of both the algorithm and the human satisfy the Mallows model), and we will study algorithms that can either rearrange the orderings of items of its ground-truth ranking $\ar^*$ or change its accuracy level $\phi_a$. 

There are various societal objectives that we may consider: in Section \ref{subsec:welfare}, we study expected utilitarian social welfare. For this objective, we show that finding the optimal arrangement of items (e.g., determining the optimal central ordering of the algorithm) is \NP-hard. Finally, even though it is computationally hard, it can be formulated as a linear mixed-integer program (MIP) with only a linear number of binary variables. However, a welfare-maximized strategy may not always be desirable in every respect: next, Section \ref{subsec:compuplift} turns to other objectives, such as uplift. We provide counterexamples where maximizing welfare violates uplift, and we also find cases where introducing a higher noise level helps achieve uplift. Finally, we provide illustrations of solutions achieving uplift for special cases (such as top-item recovery). We conclude by illustrating generalizations of our results with numerical simulations. 

\subsection{Social Welfare Maximization}\label{subsec:welfare}
First, recall that our definition of expected utilitarian social welfare is given by: 

\socialwelfare*

Next, we explore how we can maximize the utilitarian social welfare. 
Our main result is Theorem \ref{thm:SWNP}, which shows that while selecting the optimal accuracy (or noise) parameter is extremely straightforward, selecting the optimal algorithmic ordering over items can be \NP-hard.  Finding the optimal arrangement of items can be reduced to the problem of finding the maximum independent set, which is known to be \NP-complete. 
In addition, given an optimal arrangement, minimizing noise is always optimal.
However, minimizing noise is not always better for any arrangement -- especially if the arrangement is highly (equivalently, maximizing accuracy) misaligned with the ground-truth arrangements of the majority of the population.

\begin{restatable}[Computational hardness]{theorem}{thmSWNP}
\label{thm:SWNP} 
The expected social welfare is maximized when the algorithm's distribution is noiseless.
However, it remains \NP-hard to find $\ar^*$ that maximizes the expected social welfare, even in the \emph{top-item recovery} setting.
\end{restatable}

To complement the hardness result, we show that the welfare-maximizing algorithm can be exactly characterized as the optimal solution to a mixed-integer program (MIP) and solved efficiently in practice, even for $m=20$ items and heterogeneous populations of humans with $720$ types of humans.
The full proof and detailed empirical results can be found at \Cref{app:strategic_algo_welfare_max} and \Cref{app:exper_perf_of_mip}. 
Meanwhile, although finding the welfare-maximizing algorithm in the top-item recovery setting is computationally hard, one can still efficiently identify the welfare-maximizing algorithm among those that \emph{achieve uplift}, as we will show in \Cref{lem:top_recovery_comp}.

Finally, we briefly discuss the noiselessness result for the optimal arrangement in Theorem \ref{thm:SWNP}. This comes by showing that the objective can be rewritten as the sum over all possible sets of $k$ items that could be presented to the population of humans: of all of those sets, one maximizes utilitarian welfare, and thus deterministically returning that set would maximize welfare as in Definition \ref{def:socialwelfare}.

\subsection{Uplift}\label{subsec:compuplift}
However, a social welfare-maximizing strategy may not always be desired as it can end up hurting some users (i.e, they make worse decisions than they would acting by themselves).
This is especially true if a certain type of human's ground-truth ranking $\hr_i^*$ does not align well with most other humans' ground-truth rankings, and so to maximize expected social welfare, the algorithm may optimize for one type of human preference at the expense of other types of humans.
As a result, these types of humans may receive less accurate recommendations from the algorithm in the collaboration: they may even obtain lower utility than they would by themselves. A social welfare desiderata relating to this objective is \emph{uplift} across the population of people: when is it possible to design an algorithm such that every type of human benefits, relative to the utility they would obtain from solving the problem by themselves? 

\compuplift*

As a warm-up, we first consider the special case where the human is aligned with the algorithm on every item for which it has positive utility: the following lemma shows that uplift (e.g., strict benefit for the human) is guaranteed. 

\begin{restatable}{lemma}{lemAligedTopItems}
\label{lem:alignedtopitems}
Suppose human of type $i$ only has positive values for the top $T$ items, i.e., $v_{i,j} > 0$ for $j \le T$ and $v_{i, j} =0$ for $j > T$. 
Then if $\ar^*(j) = \hr^*(j)$ for any $j\le T$, $\phi_a = \phi_h$, and $1< k <m$, uplift is achieved.
\end{restatable}
Note that \Cref{lem:alignedtopitems} strictly generalizes results in \citep{donahue2024listsbetteronebenefits}, which provided these results only for when exactly $2$ items are presented and the human is in the top item recovery case (only has value for her top item).
Moreover, \Cref{lem:alignedtopitems} does not put any constraint on the zero-valued items, and the algorithm and human can be misaligned on these items.
The proof can be found at \Cref{app:uplift_and_comp}.

\subsubsection{Tension between Social Welfare and Uplift}
In the above case, a social welfare maximizing arrangement implies uplift (if feasible) since there is only one type of human.
However, when it comes to multiple misaligned humans, forcing welfare-maximization could hurt some type of human, e.g., when there are two types of humans with population shares of $0.99$ and $0.01$ and completely different preferences, maximizing welfare makes the algorithm fully aligned with the preferences of the majority type, thereby hurting the minority group.

Secondly, uplift can be viewed as a type of fairness notion: it requires that each type of human receive some benefit from using the algorithm. 
In contrast, social welfare is a well-accepted metric for evaluating the algorithmic efficiency.
In \Cref{app:tension_welfare_objectives}, we further describe numerical experiments on connections between the two objectives.
We study the change in optimal social welfare after enforcing the uplift constraint.
Usually, enforcing uplift will reduce the optimal social welfare, especially when the distribution of humans is highly imbalanced, e.g., certain types of humans dominate the population.
In addition, in \Cref{sec:empirical}, we compare the two optimal algorithmic solutions, each optimizing a single welfare objective, on a real-world dataset.
We also find that solely maximizing utilitarian welfare could cause extreme unfairness. 

\subsubsection{Small Noise Facilitates Uplift}
Given the tension between the two objectives, small noise may be an effective method to achieve uplift.
Considering the same setting as the last example, it seems likely that a small degree of randomization could sharply increase welfare for the less populous type of human.
\Cref{ex:positivetemp} shows this more formally, and the proof can be found at \Cref{app:noiseness_help_uplift}.

\begin{restatable}{lemma}{lemPostiveTemp}
\label{ex:positivetemp}
There exist settings where uplift can occur at lower accuracy (higher noise), but fails at higher accuracy. 
\end{restatable}

We remark that \Cref{ex:positivetemp}'s message is in line with other works on fairness (e.g., \citep{jain2024scarce, singh2018fairness}) that deem randomization to be helpful for fairness. One difference in our setting is that we expect that in many cases there will be relatively large benefits in terms of uplift from relatively small levels of noise (and thus, relatively small impact on the utilitarian welfare objective). This is because small amounts of noise that result in occasionally including a new item has \emph{asymmetric costs and benefits}: it has large benefits for users who stronger prefer that item to those that are already included because they will be able to select it with relatively high probability. However, it has small negative impacts on users who have low value for that item, since on the occasions when it is included, they will have a high probability of correctly ignoring that item as irrelevant. 

\subsubsection{General Results for Uplift}
\label{subsubsec:general_results_for_uplift}
Motivated by the above examples, we consider how to design an algorithm that meets the objective of uplift.
We begin by noting that there are cases where achieving uplift is impossible. To build intuition, consider a scenario in which humans are highly misaligned in their top preferences. In such situations, an uplift strategy is only feasible if all of their most preferred items can be accommodated within the limited window of size $k$. This naturally motivates characterizing instances where uplift is achievable. To this end, we first show that it remains \NP-hard to determine the existence of a strategy (which may be noisy) that achieves uplift.
However, we note that when the size of presented items $k$ is constant, one can verify whether a given strategy achieves uplift. Thus, for a constant $k$, the problem is in \NP. We note that in real-life settings, we expect $k$ to be relatively small, given that humans would likely be unable to process very large subsets of items. 
The proof (can be found at \Cref{app:np_hard_for_deciding_uplift_existence}) is based on a reduction from the vertex cover problem. 
\begin{restatable}{theorem}{thmCompGeneral}
\label{thm:comp_general}
It is \NP-hard to determine whether there exists a strategy satisfying uplift $(\ar^*, \phi_a)$.
Further, it is \NP-complete when $k$ is given as a constant (which means whether a given strategy $(\ar^*, \phi_a)$ achieves uplift can be verified in polynomial time).
\end{restatable}

Next, we provide characterizations of what an uplift strategy looks like in some special cases, such as top item recovery (where we show that the optima algorithm is a noiseless algorithm and presents only items that at least one human values), and where the humans share the same set of best items. 

For the top item recovery setting, consider a simple strategy that ranks all these top items first in the algorithm's ranking and sets the noise as zero (equivalently, set $\phi_a$ as $+\infty$).
\Cref{lem:top_recovery_comp} shows that the simple strategy achieves uplift. 
We note that maximizing social welfare remains \NP-hard even in this special case (using the same reduction of~\Cref{thm:SWNP}).

\begin{restatable}{lemma}{lemTopRecoveryCom}
\label{lem:top_recovery_comp}
In the top item recovery setting, when the ground-truth rankings of humans satisfy $\abs{\mathcal{M}_0} \le m-1$ where $\mathcal{M}_0$ is the set of distinct items that some human has positive value for, then always presenting $\mathcal{M}_0$ to the human achieves uplift and also maximizes the expected utilitarian social welfare among all the noiseless algorithms that achieve uplift.
\end{restatable}

\section{Empirical Study}\label{sec:empirical}
In this section, we conduct an empirical study of human-AI collaboration with misaligned ground-truth rankings on a real-world dataset to answer the following question: \emph{how does human accuracy affect the tension between different welfare objectives?}

We use the sushi preference dataset~\citep{kamishima2003nantonac}, which consists of about 5k rankings of ten kinds of sushi, where we
denote the ten types of sushi by $x_1, \dots, x_{10}$. For computational efficiency, we consider only the partial rankings of the humans in our sushi dataset over the items $x_1, \ldots, x_5$, treating these partial rankings as their ground-truth preferences.
Each ranking is also associated with the fraction of the population in the dataset who had this preference over the types of sushi. A detailed breakdown of the most frequent sushi rankings is presented in \Cref{tab:tablesec6} (\Cref{app:tablesec6}).

Note that in this dataset, it is inherently impossible to distinguish between a) population-level heterogeneity in preferences, and b) user-level noise in reporting those preferences. To our knowledge, there is not a preference dataset that cleanly displays both a) and b). As such, in this section we make the assumption that users have reported their sushi preferences perfectly (no errors from b) and that the preference dataset simply reflects user heterogeneity (only difference is from a). In order to model user-level noise, we add noise (e.g., Mallows) on top of the ground-truth preferences within \citep{kamishima2003nantonac}. Furthermore, because this dataset does not include explicit human valuation scores for the sushi items, we assume item values within each ranking is determined by their \emph{Borda counts}. Formally, given a ranking $\pi$ over $m$ items $\{x_1, \ldots, x_m\}$, where $\pi(i)$ denotes the position of item $x_i$ (with $\pi(i) = 1$ being the most preferred), the Borda score assigned to item $x_i$ is
\[
B(x_i; \pi) = m - \pi(i).
\]
Thus, the top-ranked item receives a score of $m - 1$, the second-ranked item a score of $m - 2$, and so on, with the least-preferred item receiving a score of $0$. This provides a simple cardinal approximation of ordinal preferences, which we use as the value function in our experiments.

The objective is to design an algorithm that presents only three items from $x_1, \ldots, x_5$ to each human to maximize a chosen notion of welfare. We find this experiment especially natural for our setting: if a restaurant created a \enquote{menu} of options that is quite large ($>100$), a user would probably have a low probability of finding sushi she likes. However, if the restaurant creates a very small menu, then user-level heterogeneity in preferences might mean that some humans are hurt by the algorithm. In this section, we explore these trade-offs within the specific sushi preference dataset. We assume the algorithm is perfectly noiseless, so the menu always consists of the top three sushis in its ground-truth ranking.
To study the tension between different welfare objectives, we compare three different welfare objectives:
\begin{enumerate}[leftmargin=*]
\item Let $A_w$ denote the algorithm that maximizes \emph{utilitarian welfare}.
\item Let $A_m$ denote the \emph{majority ranking}, i.e., the ranking supported by the largest number of individuals.
\item Let $A_u$ denote the algorithm that maximizes \emph{uplift}, defined as the number of humans who benefit from collaboration.
\end{enumerate}
\begin{figure}[ht]
\begin{minipage}[t]{0.49\linewidth}
\centering
    \includegraphics[width=0.9\textwidth]{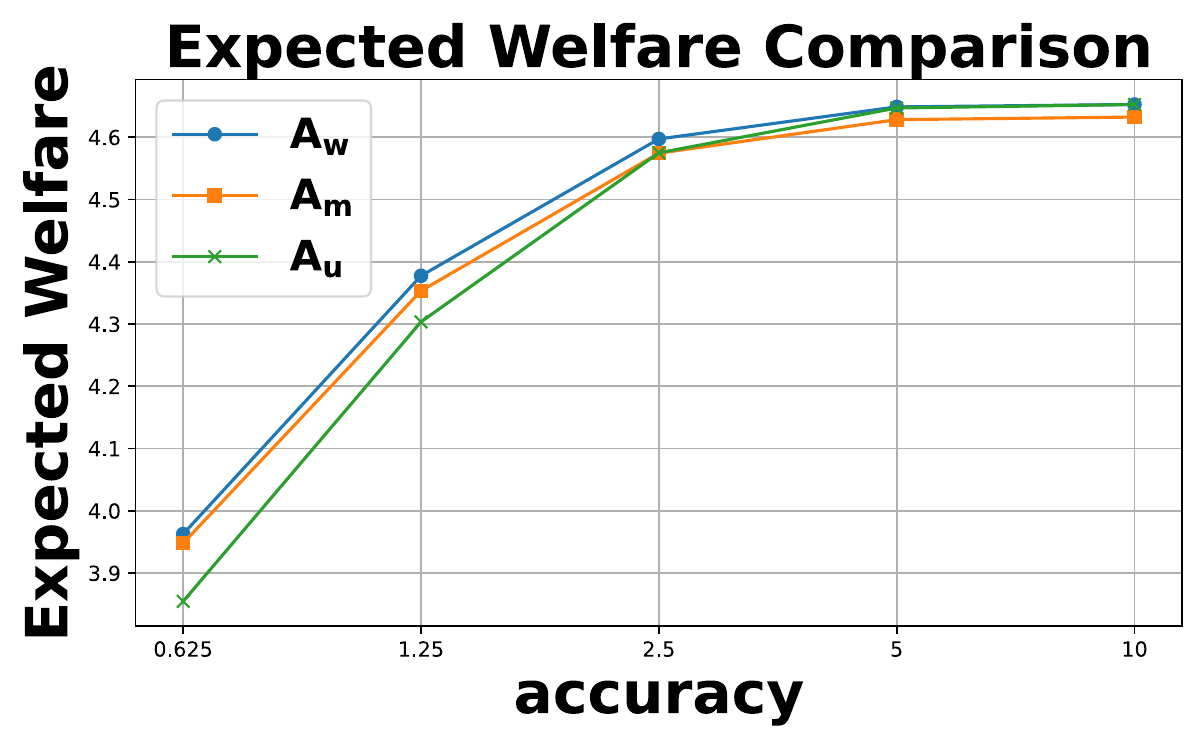}
    \end{minipage}
\hfill
    \begin{minipage}[t]{0.49\linewidth}
\centering    \includegraphics[width=0.9\textwidth]{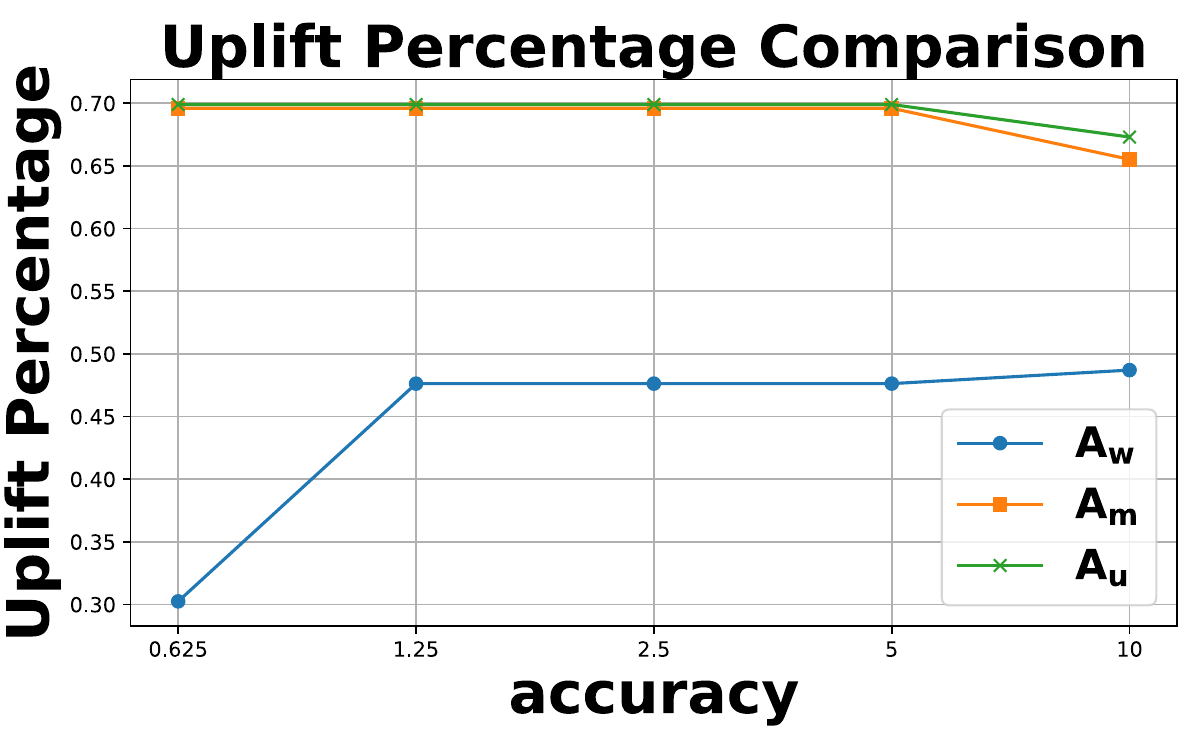}
\end{minipage}
        \caption{Utilitarian welfare and uplift achieved by the majority algorithm $A_m$, the welfare-maximizing algorithm $A_w$, and the uplift-maximizing algorithm $A_u$ across varying levels of human accuracy.}\label{fig:comparison_a_tb}
\end{figure}

In \Cref{fig:comparison_a_tb}, we plot the utilitarian welfare and uplift achieved by each of the algorithms across different levels of human accuracy. In the left plot, we observe that as humans become more accurate, their expected utilitarian welfare increases under both $A_m$ and $A_w$. However, this monotonic improvement does not hold for $A_u$: beyond a certain human accuracy threshold, the performance of this algorithm declines.

Interestingly, in the right plot, we observe that as humans become more accurate, after a certain threshold, all algorithms can guarantee positive uplift for a smaller fraction of individuals. This suggests that when human judgments are noisier, a larger proportion of people benefit from collaboration with a fixed algorithm. However, as humans become more accurate, no single algorithm remains universally beneficial. Moreover, we observe that the algorithm optimizing utilitarian welfare performs substantially worse than the other two in terms of uplift across most accuracy regimes. This suggests a trade-off between the uplift and utilitarian welfare objectives. In this dataset, the majority algorithm performs relatively well on both metrics.

\section{Discussion}
\label{sec:discuss}
In this paper, we studied a model of human-algorithm collaboration in a setting where the algorithm acts as a \emph{curator} of items for a \emph{noisy} human, and the algorithm may be \emph{misaligned} with the human's preferences over items. In particular, we define misalignment as a disagreement on the relative ordering over items by value: while such misalignment can be detrimental (because the algorithmic assistant returns items the human may find suboptimal), it can also be \emph{beneficial} (because the assistant makes \enquote{different mistakes} that the human may be more easily able to correct). Our goal in this paper was to analyze when the relative benefits of misalignment outweigh the costs. 

In Section \ref{sec:whobenefitsmore} we began by studying the preferences an \emph{individual} human would have over an algorithmic assistant, showing that certain types of misalignment are helpful. In Section \ref{sec:strathum}, we turned to the question of designing an algorithm to maximize welfare over a \emph{population} of humans who have diverse preferences. While in general it is not possible to simultaneously maximize the welfare of each human engaging with the algorithm, we study how to maximize some collective welfare objectives, like maximizing utilitarian welfare or ensuring that each human strictly benefits from using the algorithmic assistant, relative to what utility she could get by herself. In general, we show that exactly achieving either of these objectives is computationally hard. However, we also give positive results for restricted settings: For maximizing utilitarian welfare, we present a mixed-integer program, building on similar efficient formulations previously applied to problems involving Mallows preferences. In addition, we identify compelling special cases where uplift is attainable and provide an explicit algorithm that achieves it.
Throughout, we supplement our theoretical results with numerical simulations, exploring generalizations in permutation distributions as well as values agents have for each item. 



\section{Extensions}\label{sec:extensions}
There are numerous potential extensions to our work.
In \Cref{sec:model_extension}, we describe a few extensions of our theoretical model, including the extensions of the noisy model and the way of human-algorithm interaction.
In \Cref{sec:implications}, we discuss substantial implications for policy making, specifically on the relative benefits and harms of misalignment for users.

\subsection{Model Extensions}
\label{sec:model_extension}
First, some extensions could involve relaxing our theoretical model. For example, our main theoretical results focus on permutation models related to the number of inversions (e.g., the Mallows model). While we show in numerical simulations that our core results generalize to other models such as RUM, extending these results theoretically (or extending our results to other classes of permutations) could be interesting. Additionally, one core assumption of our work is that humans have the same \emph{magnitude of preferences}: that is, every human has the same value for their favorite item, the same value for their second-favorite item, and so on (even though the identity of their favorite items may differ). This corresponds to assuming that all humans are equally \enquote{picky}. While this assumption is necessary for tractability purposes, this is likely not true in general: some humans may be equally happy with multiple options, while other humans may be much more selective. Relaxing this assumption may lead to different trade-offs in objectives. 

Other extensions could change our model of human-algorithm interaction. For example, our model of subset curation assumes that the human observes the set of presented items as an unordered \emph{set} and is not biased by any inherent ordering of items. This is almost certainly not true in practice, as it has been demonstrated that humans are often biased by the order in which items are presented (e.g., see \citep{joachims2007evaluating, agichtein2006improving, craswell2008experimental}), and this effect has been studied in human-algorithm collaboration more specifically (e.g. \citep{mclaughlin2024designing, donahue2024listsbetteronebenefits, wang2022modeling}). It could also be possible that the human and algorithm undergo a sequence of back-and-forth interactions which ultimately end up at a final solution (as in \citep{collina2025collaborative}), which may lead to different implications on personalization. Finally, our results on optimizing the algorithm for different social welfare functions assume that the algorithm has perfect knowledge of the distribution of human preferences. In practice, this is almost certainly not the case, and it could be useful to study settings where the algorithm must learn human preferences over multiple interactions. 

\subsection{Implications for Policy}
\label{sec:implications}
Our results have potential implications for policymakers, designers of algorithms, and users of algorithmic tools. 

\paragraph{Potential for benefits in misalignment.} In some real-life settings, humans have near-perfect ability to identify the \enquote{correct} outcome (e.g., the political opinion that they most closely identify with). However, in many real-life settings, humans may only be able to imperfectly (noisily) recover their own true preferences: prior research has shown that humans noisily pick the favorite item from sets, such as presented movies or items on a menu \citep{agranov2017stochastic, plackett1975analysis, luce1959individual, hey1994investigating}. In this setting, our results show that having algorithms that are \enquote{misaligned} in specific ways can be beneficial. This implies that designers of algorithms should maybe not always focus on designing algorithms that personalize to individual humans (eliminating misalignment): in settings where humans are noisy, personalizing may end up being harmful. This implies that algorithmic designers may want to take different strategies towards alignment depending on the context in which humans will use the tool, focusing on alignment more in noiseless settings and less in settings where humans are expected to know their own preferences less well. 

\paragraph{Benefits of positive temperature.} Our results also show that there can be benefits to algorithmic noise (e.g., positive temperature) when a single algorithm is serving a population of heterogeneous users. Benefits of randomization for goals related to fairness have been shown in other settings (e.g., \citep{jain2024scarce, singh2018fairness}), though to our knowledge, this question is less well studied in the context of complementarity or similar benchmarks. Most generative AI tools have a constant temperature across multiple queries: our work suggests that for certain welfare objectives, it could be helpful to have temperature change dynamically, being lower in settings where human opinions are roughly in agreement and larger in settings where human opinions are more diverse. 

\paragraph{Implications for users.} Our results also suggest strategies for users of algorithmic tools. Users sometimes have direct and indirect ways of controlling personalization: for example, they can choose to accept or delete cookies, can opt to use incognito mode, can volunteer or (sometimes) delete information generative tools have stored on their prior interactions, or can select tools that they believe to be more or less closely aligned with their beliefs. Our results provide motivation for users to sometimes strategically seek options that reduce personalization, even absent other concerns such as privacy. 

\paragraph{Tensions in societal goals.} It should not be surprising that some societal goals may be in tension with each other. However, we hope that our research could highlight particular objectives that are relatively understudied in the pluralistic alignment setting, such as ensuring all types of humans benefit in a randomized task (as compared to their utility if they solved the problem themselves). While we believe that there could be contexts in which either utilitarian welfare or goals like uplift would be optimal, we think it is worth having designers deliberately decide on which they may choose to optimize in different contexts. 

\section{Numerical Simulations}
\subsection{Numerical Extension of \cref{thm:benefit_inv1}}
\label{app:extension_of_theorem1}

We consider $m=4$ items, of which the algorithm displays $k=2$.
For the human, the value of item $x_i$, $v_i\propto e^{-\beta\cdot j}$, where $\beta\ge0$ controls the heterogeneity of values. 
Heterogeneity helps to relax the top item recovery case smoothly.

\paragraph{Extensions to Mallows Model.}
The first experiment assumes that both human and algorithm rankings follow a Mallows model with the same accuracy of $\phi_a = \phi_h = 0.5$.
Let $\hr^* = (x_1,x_2,x_3,x_4)$ be the human's ground-truth ranking, and let $\ar^*$ be an arbitrary algorithm's ground-truth ranking.  
\Cref{fig:mallows_aligned_misaligned} plots the \emph{difference} in human's expected utility working with various misaligned algorithms, and the aligned algorithm one. The right figure \Cref{fig:mallows_aligned_misaligned} restricted to top-aligned algorithms (algorithms that place $x_1$ first).

\begin{figure}[h]
    \centering
    \includegraphics[width=0.7\linewidth]{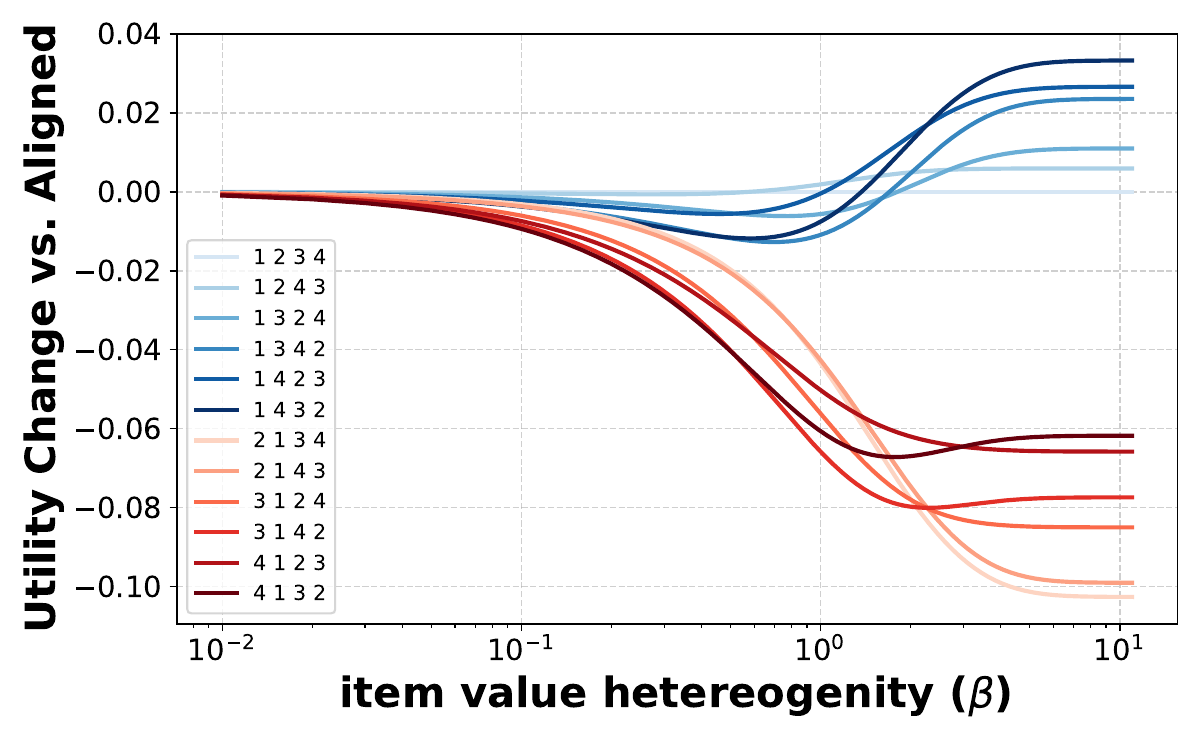}
    \includegraphics[width=0.7\linewidth]{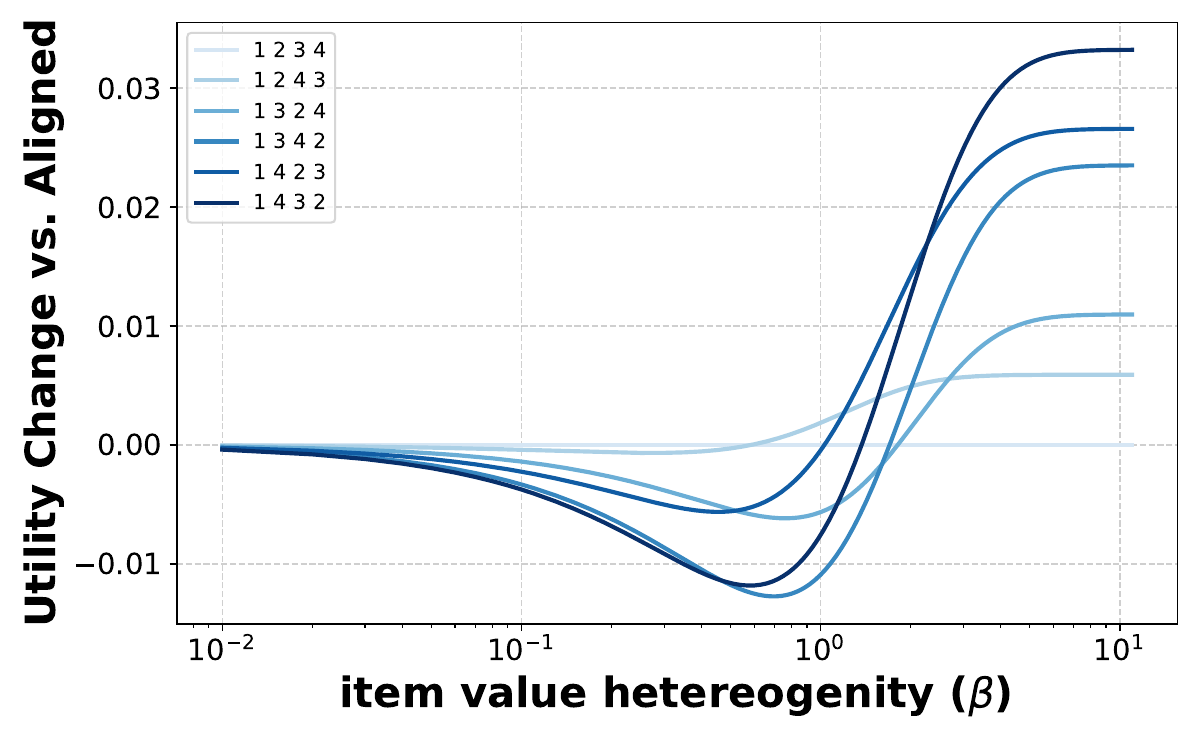}
    \caption{Comparison of human's expected utility differences after collaboration with a misaligned and an aligned algorithm, as a function of $\beta$. Each curve represents an algorithm with the corresponding ground-truth ranking.}
    \label{fig:mallows_aligned_misaligned}
\end{figure}

For small $\beta$, valuations are nearly uniform, so the human's utility is similar across all algorithms. 
As $\beta$ increases, the aligned algorithm performs best, but when $\beta$ becomes sufficiently large, the utility is dominated by the top item $x_1$, and top-aligned algorithms yield higher utility. 
In the extreme top-item recovery case (\Cref{thm:topitem}), the most inverted top-aligned algorithm achieves the highest utility, explaining why the curve for $\ar^*= (x_1,x_4,x_3,x_2)$ dominates at large $\beta$. 
Consistent with \Cref{thm:benefit_inv1}, all top-aligned algorithms outperform the aligned one when $\beta$ is large. 

\paragraph{Extensions to Random Utility Model.}
We perform a similar simulation by assuming that both the human and the algorithm follow the Plackett-Luce RUM with Gumbel noises of $0.1$. 
In \Cref{fig:RUM-all-misaligned}, we observe that top-aligned algorithms quickly outperform others. This observation mirrors the second property established in \Cref{thm:benefit_inv1}: swapping item $x_1$ from the top position with any lower-ranked item yields a new algorithm ranking that is less advantageous for the human.

We emphasize an important distinction in the RUM setting: unlike the Mallows model, modifying item values alters the probabilities of sampling different rankings. 
As a result, when $\beta$ becomes large, the probability that top-aligned algorithms present item~$x_1$, as well as the probability that the human selects item~$x_1$, both approach 1. This convergence explains why, at high $\beta$ values, the expected utilities of all top-aligned algorithms and the aligned algorithm become nearly identical in RUM.

\begin{figure}[h]
    \centering
    \includegraphics[width=0.7\linewidth]{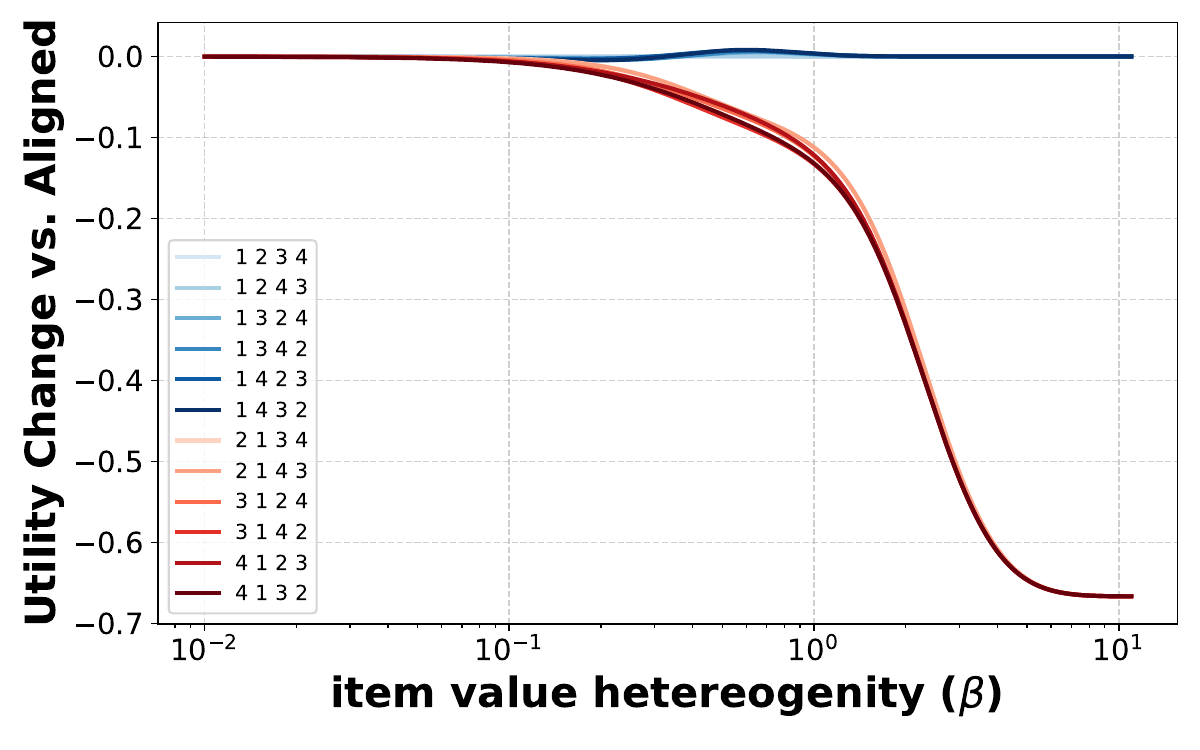}
    \includegraphics[width=0.7\linewidth]{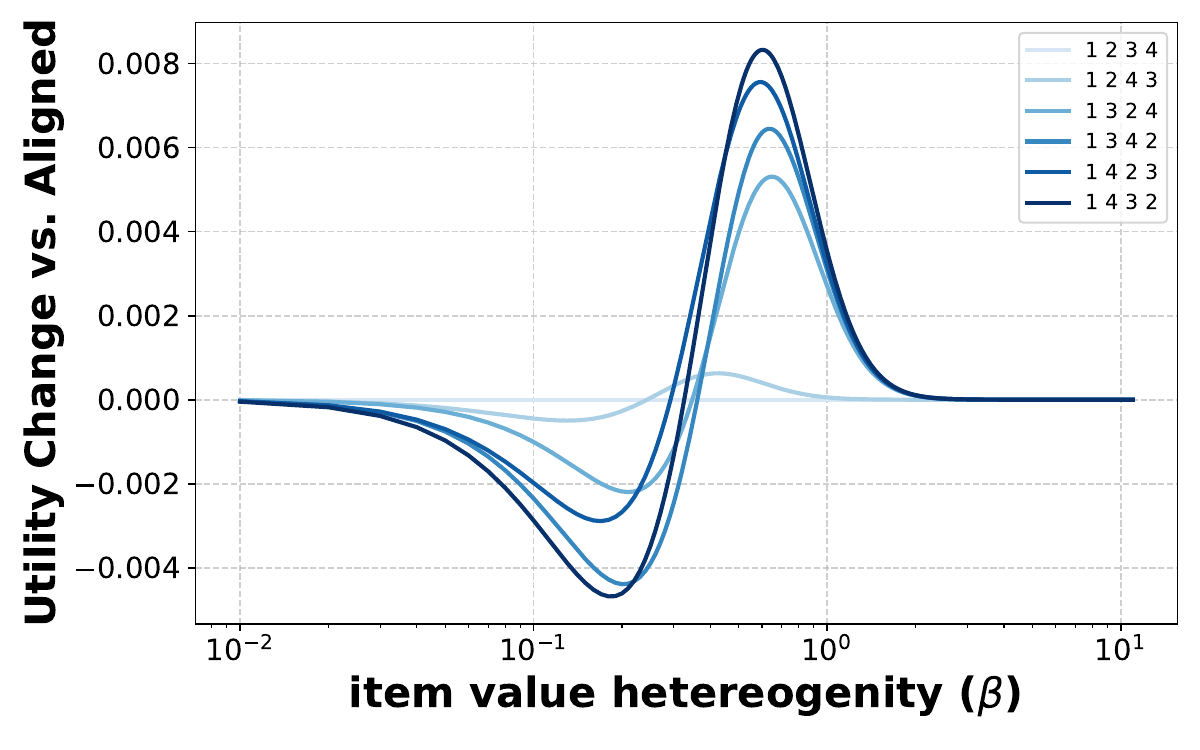}
    \caption{Comparison of human's expected utility differences under RUM.}
    \label{fig:RUM-all-misaligned}
\end{figure}

\begin{takeaway} 
In the Mallows model and RUM, misalignment on relatively small-valued items still helps the human.
\end{takeaway}

\subsection{Performance of MIP}
\label{app:exper_perf_of_mip}
We evaluate the performance of the MIP by varying both the number of items $m$ and the number of types of humans $n$.
In the first setting, we consider two humans with ground-truth rankings $\hr_1^*=(x_1,x_2,x_3,x_4,\ldots,x_m)$ and $\hr_2^* =(x_4,x_2,x_3,x_1,\ldots,x_m)$, where only $x_1$ and $x_4$ are swapped, and utilities $(v_{i, j})_{j=1}^m={4,3,2,1,0\ldots,0}$.
We vary $m$ from $10$ to $70$ and test $k=2,4,6,8$. 
For each test, we vary the frequencies of the first type of human from $0$ to $1$ and take the average of the running time.
In the second setting, humans are drawn from a Mallows distribution with $n=t!$ possible types, $t=1,\dots, 6$.
We fix $m=20$, and test $k=2,4,6,8$.
As shown in \Cref{fig:mip_time}, all instances terminate within 175 seconds, with the slowest cases taking 40.46 and 169.92 seconds, respectively. 
The running time grows polynomially in $m$ and linearly in $n$.

\begin{figure}[h]
\begin{minipage}{0.49\linewidth}
\centering
\includegraphics[width=0.8\textwidth]{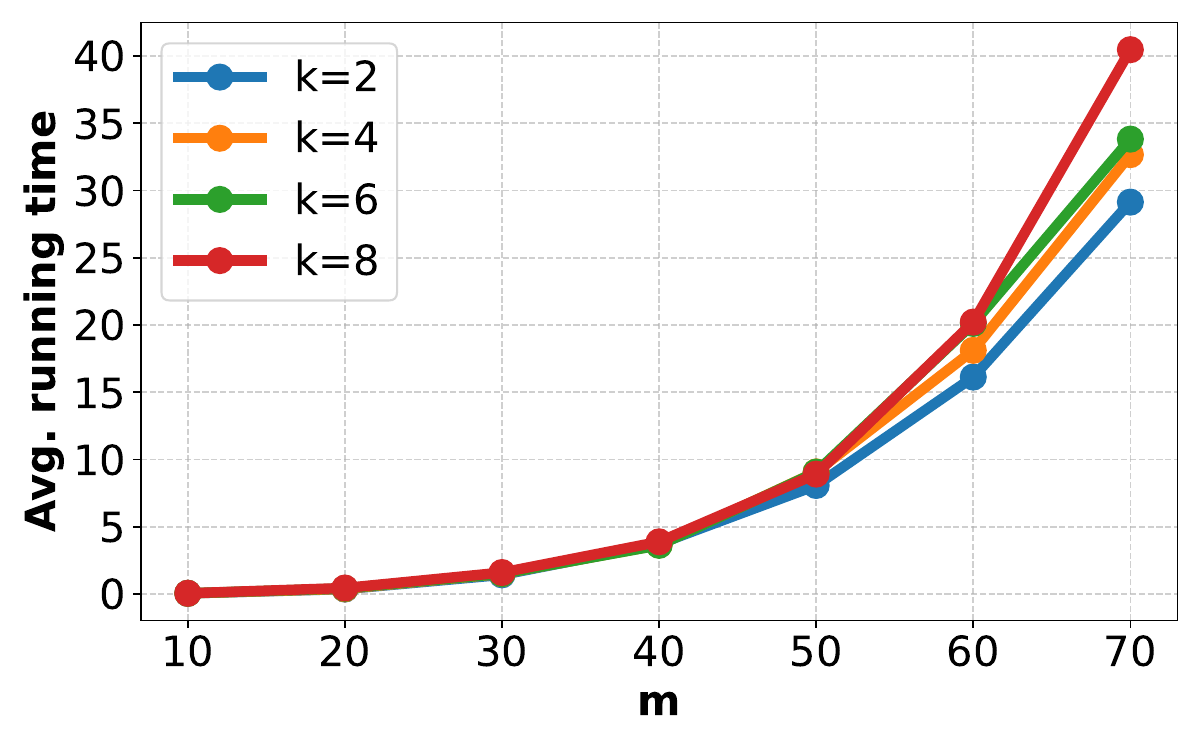}
\end{minipage}
\hfill
\begin{minipage}{0.49\linewidth}
\centering
\includegraphics[width=0.8\textwidth]{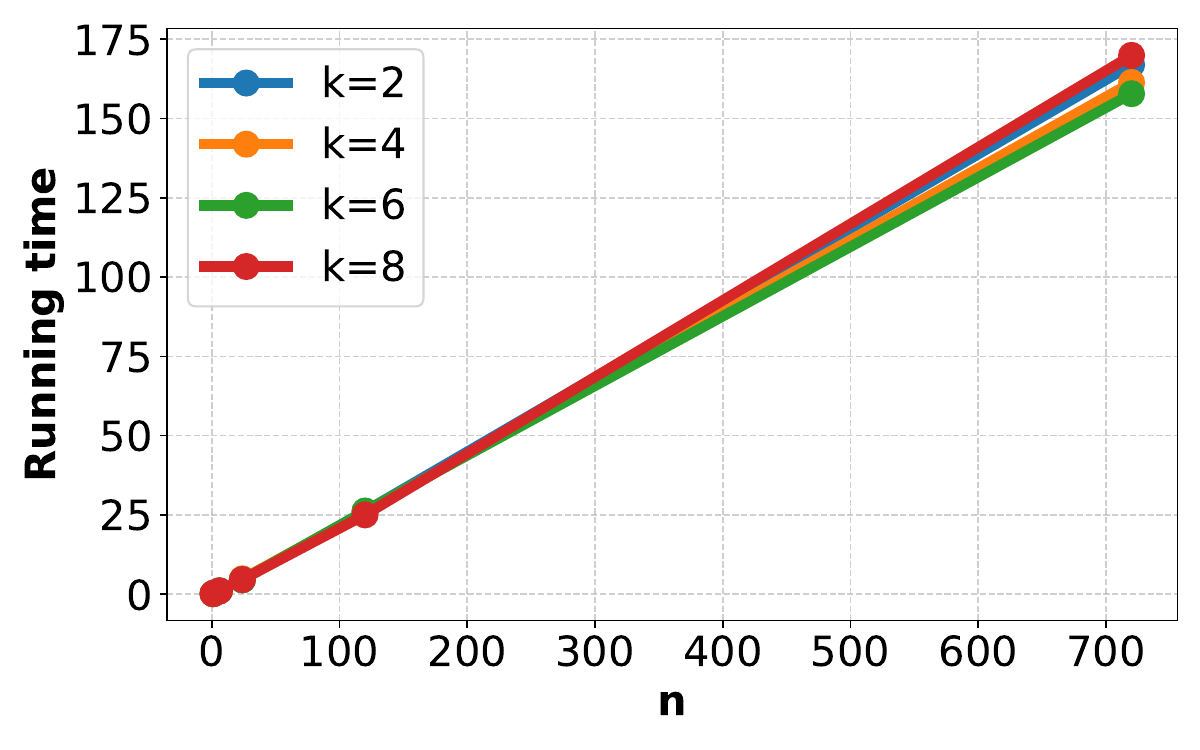}
\end{minipage}
\caption{Running time of MIP (varying $m$ and $n$)}
\label{fig:mip_time}    
\end{figure}

\begin{takeaway} 
The welfare-maximizing algorithm can be found efficiently in practice.
\end{takeaway}

\subsection{Tension between Welfare Objectives}
\label{app:tension_welfare_objectives}

The part mainly discusses the effect of enforcing uplift on social welfare. 
We consider $m=6$ items and the humans’ preferences differ only in the top three items, which follow a Mallows distribution with parameter $\gamma$, reflecting the balance of the populations.
When $\gamma$ is large, most humans have preferences aligned with $(x_1, x_2, \ldots, x_t)$.
When $\gamma$ is small, the humans' preferences become more dispersed.
Each human’s utilities are set as $(v_{i, j})_{j=1}^6 = {1, 1, 0.5, 0.2, 0, 0}$.
All humans share the same accuracy parameter $\phi_h$, which we vary from $0$ to $3$ in increments of $0.3$.
\Cref{fig:mip_uplift_t_3} plots the social welfare of the welfare-maximizing algorithms with and without the uplift constraint under different values of $\gamma$.
The expected social welfare under the uplift constraint is consistently lower than that without it.
When human preferences become more concentrated ($\gamma = 3$), the optimal solutions start to diverge.
The largest welfare gap occurs at $\phi_h = 1$, where the expected social welfare without the uplift constraint is $0.989$, compared to $0.954$ with the constraint.

\begin{figure}[h]
\begin{minipage}{0.49\linewidth}
\centering
\includegraphics[width=0.8\textwidth]{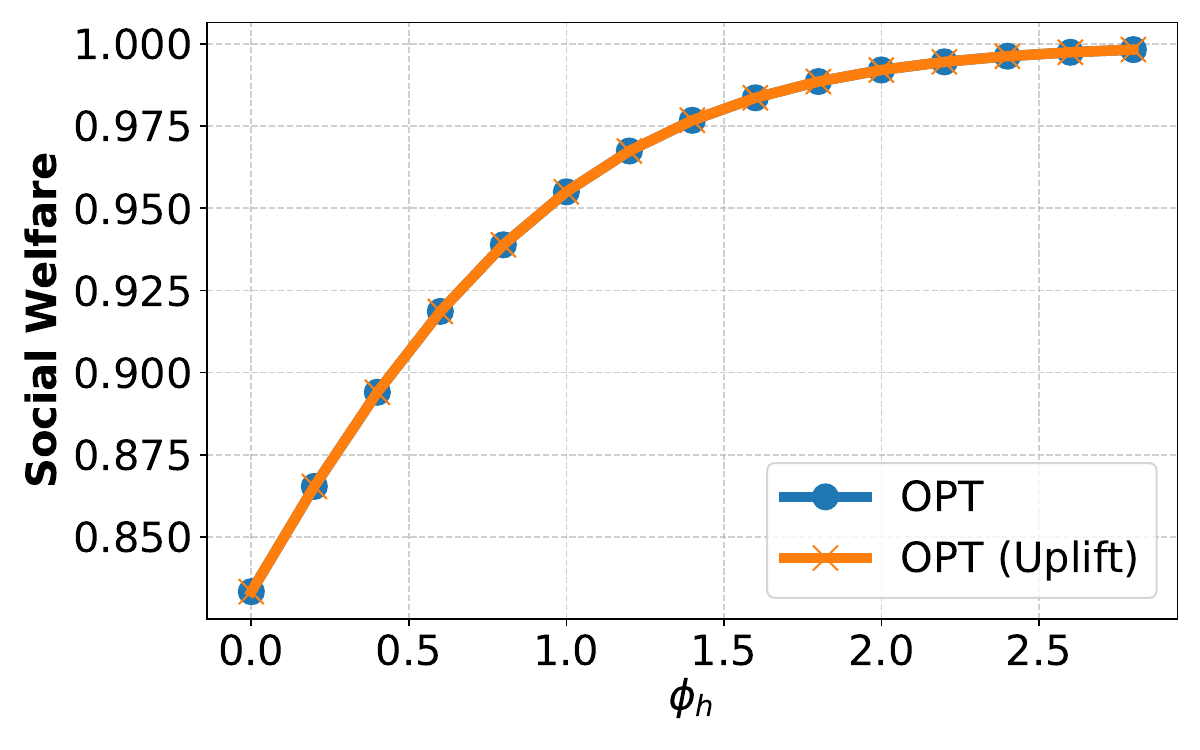}
\end{minipage}
\hfill
\begin{minipage}{0.49\linewidth}
\centering
\includegraphics[width=0.8\textwidth]{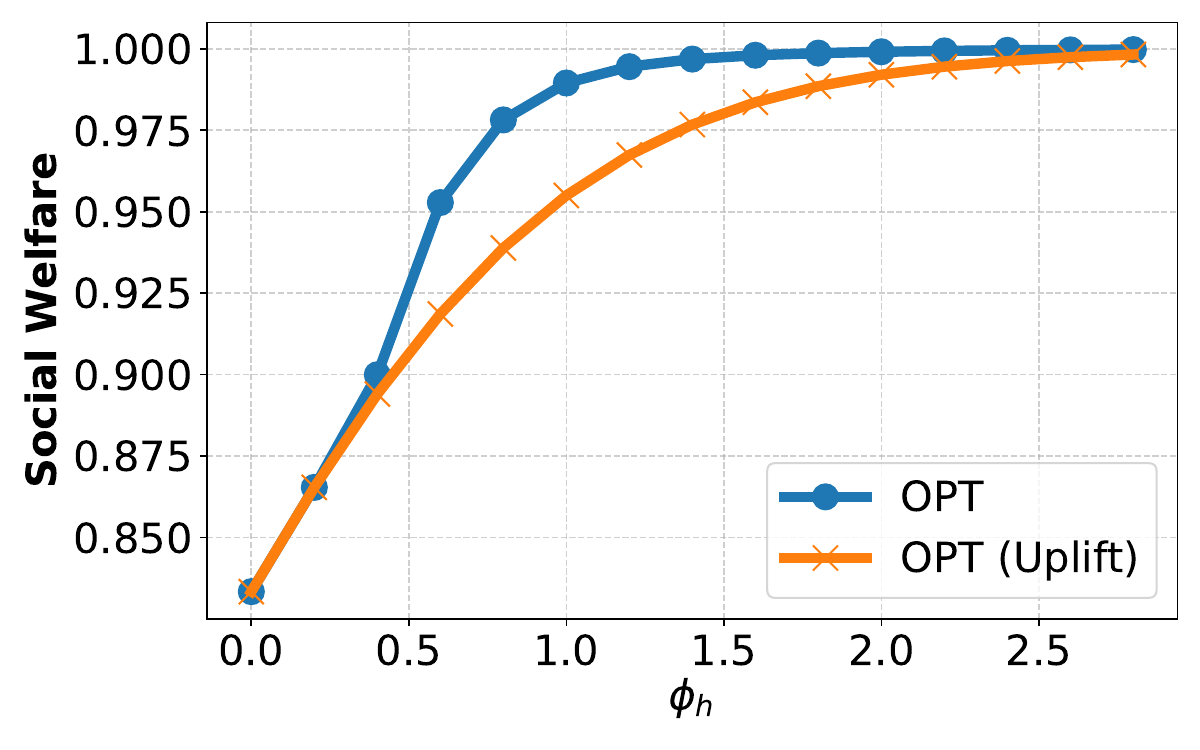}
\end{minipage}
\caption{Max welfare: $\pm$Uplift of $ \gamma=0.5$ (left) and $3$ (right)}
\label{fig:mip_uplift_t_3}
\end{figure}    

\begin{takeaway} 
Enforcing uplift reduces the optimal social welfare, especially when humans' distribution is imbalanced.
\end{takeaway}

\section*{Acknowledgements}
The research of Bhaskar Ray Chaudhury and Jiaxin Song was supported by NSF CAREER Grant CCF-2441580.
The research of Kate Donahue was supported by the MIT METEOR Fellowship.
The research of Parnian Shahkar was supported by NSF Grant CCF-2454115.
We thank Kira Goldner (Boston University), Michael Kearns (University of Pennsylvania), Ariel Procaccia (Harvard University), Aaron Roth (University of Pennsylvania), Jann Spiess (Stanford), Kostas Kollias (Google Research), Manish Raghavan (MIT), Biaoshuai Tao (SJTU), Ming Yin (Purdue), as well as attendees at the Themis Science Talk at Amazon, the Human-AI workshop at EC'25, the Human-AI Complementarity for Decision Making Workshop at CMU, the AI and Science Group at Meta, and the Boston Economics and Computing Hub Workshops, for their valuable discussions and insightful suggestions.
We also thank the anonymous reviewers for their constructive feedback.

\bibliographystyle{plainnat}
\bibliography{reference}

@misc{donahue2024listsbetteronebenefits,
      title={When Are Two Lists Better than One?: Benefits and Harms in Joint Decision-making}, 
      author={Kate Donahue and Sreenivas Gollapudi and Kostas Kollias},
      year={2024},
      eprint={2308.11721},
      archivePrefix={arXiv},
      primaryClass={cs.LG},
      url={https://arxiv.org/abs/2308.11721}, 
}

@article{arnaiz2025towards,
  title={Towards Human-AI Complementarity in Matching Tasks},
  author={Arnaiz-Rodriguez, Adrian and Benz, Nina Corvelo and Thejaswi, Suhas and Oliver, Nuria and Gomez-Rodriguez, Manuel},
  journal={arXiv preprint arXiv:2508.13285},
  year={2025}
}

@article{plackett1975analysis,
  title={The analysis of permutations},
  author={Plackett, Robin L},
  journal={Journal of the Royal Statistical Society Series C: Applied Statistics},
  volume={24},
  number={2},
  pages={193--202},
  year={1975},
  publisher={Oxford University Press}
}

@book{luce1959individual,
  title={Individual choice behavior},
  author={Luce, R Duncan and others},
  volume={4},
  year={1959},
  publisher={Wiley New York}
}

@inproceedings{kamishima2003nantonac,
  title={Nantonac collaborative filtering: recommendation based on order responses},
  author={Kamishima, Toshihiro},
  booktitle={Proceedings of the ninth ACM SIGKDD international conference on Knowledge discovery and data mining},
  pages={583--588},
  year={2003}
}

@article{fontana2023conformal,
  title={Conformal prediction: a unified review of theory and new challenges},
  author={Fontana, Matteo and Zeni, Gianluca and Vantini, Simone},
  journal={Bernoulli},
  volume={29},
  number={1},
  pages={1--23},
  year={2023},
  publisher={Bernoulli Society for Mathematical Statistics and Probability}
}

@article{angelopoulos2023conformal,
  title={Conformal prediction: A gentle introduction},
  author={Angelopoulos, Anastasios N and Bates, Stephen and others},
  journal={Foundations and Trends{\textregistered} in Machine Learning},
  volume={16},
  number={4},
  pages={494--591},
  year={2023},
  publisher={Now Publishers, Inc.}
}

@article{hey1994investigating,
  title={Investigating generalizations of expected utility theory using experimental data},
  author={Hey, John D and Orme, Chris},
  journal={Econometrica: Journal of the Econometric Society},
  pages={1291--1326},
  year={1994},
  publisher={JSTOR}
}

@article{agranov2017stochastic,
  title={Stochastic choice and preferences for randomization},
  author={Agranov, Marina and Ortoleva, Pietro},
  journal={Journal of Political Economy},
  volume={125},
  number={1},
  pages={40--68},
  year={2017},
  publisher={University of Chicago Press Chicago, IL}
}

@article{sorensen2024roadmap,
  title={A roadmap to pluralistic alignment},
  author={Sorensen, Taylor and Moore, Jared and Fisher, Jillian and Gordon, Mitchell and Mireshghallah, Niloofar and Rytting, Christopher Michael and Ye, Andre and Jiang, Liwei and Lu, Ximing and Dziri, Nouha and others},
  journal={arXiv preprint arXiv:2402.05070},
  year={2024}
}

@inproceedings{bansal2021does,
  title={Does the whole exceed its parts? the effect of ai explanations on complementary team performance},
  author={Bansal, Gagan and Wu, Tongshuang and Zhou, Joyce and Fok, Raymond and Nushi, Besmira and Kamar, Ece and Ribeiro, Marco Tulio and Weld, Daniel},
  booktitle={Proceedings of the 2021 CHI conference on human factors in computing systems},
  pages={1--16},
  year={2021}
}

@article{desir2016assortment,
  title={Assortment optimization under the mallows model},
  author={D{\'e}sir, Antoine and Goyal, Vineet and Jagabathula, Srikanth and Segev, Danny},
  journal={Advances in Neural Information Processing Systems},
  volume={29},
  year={2016}
}

@inproceedings{zhang2024evaluating,
  title={Evaluating the utility of conformal prediction sets for ai-advised image labeling},
  author={Zhang, Dongping and Chatzimparmpas, Angelos and Kamali, Negar and Hullman, Jessica},
  booktitle={Proceedings of the 2024 CHI Conference on Human Factors in Computing Systems},
  pages={1--19},
  year={2024}
}

@article{hullman2025conformal,
  title={Conformal prediction and human decision making},
  author={Hullman, Jessica and Wu, Yifan and Xie, Dawei and Guo, Ziyang and Gelman, Andrew},
  journal={arXiv preprint arXiv:2503.11709},
  year={2025}
}

@article{guo2025value,
  title={The Value of Information in Human-AI Decision-making},
  author={Guo, Ziyang and Wu, Yifan and Hartline, Jason and Hullman, Jessica},
  journal={arXiv preprint arXiv:2502.06152},
  year={2025}
}

@article{ackley1985learning,
  title={A learning algorithm for Boltzmann machines},
  author={Ackley, David H and Hinton, Geoffrey E and Sejnowski, Terrence J},
  journal={Cognitive science},
  volume={9},
  number={1},
  pages={147--169},
  year={1985},
  publisher={Elsevier}
}

@article{collina2025collaborative,
  title={Collaborative Prediction: Tractable Information Aggregation via Agreement},
  author={Collina, Natalie and Globus-Harris, Ira and Goel, Surbhi and Gupta, Varun and Roth, Aaron and Shi, Mirah},
  journal={arXiv preprint arXiv:2504.06075},
  year={2025}
}

@article{wang2022modeling,
  title={Modeling and correcting bias in sequential evaluation},
  author={Wang, Jingyan and Pananjady, Ashwin},
  journal={arXiv preprint arXiv:2205.01607},
  year={2022}
}

@article{mclaughlin2024designing,
  title={Designing Algorithmic Recommendations to Achieve Human-AI Complementarity},
  author={McLaughlin, Bryce and Spiess, Jann},
  journal={arXiv preprint arXiv:2405.01484},
  year={2024}
}

@inproceedings{craswell2008experimental,
  title={An experimental comparison of click position-bias models},
  author={Craswell, Nick and Zoeter, Onno and Taylor, Michael and Ramsey, Bill},
  booktitle={Proceedings of the 2008 international conference on web search and data mining},
  pages={87--94},
  year={2008}
}

@inproceedings{agichtein2006improving,
  title={Improving web search ranking by incorporating user behavior information},
  author={Agichtein, Eugene and Brill, Eric and Dumais, Susan},
  booktitle={Proceedings of the 29th annual international ACM SIGIR conference on Research and development in information retrieval},
  pages={19--26},
  year={2006}
}

@article{joachims2007evaluating,
  title={Evaluating the accuracy of implicit feedback from clicks and query reformulations in web search},
  author={Joachims, Thorsten and Granka, Laura and Pan, Bing and Hembrooke, Helene and Radlinski, Filip and Gay, Geri},
  journal={ACM Transactions on Information Systems (TOIS)},
  volume={25},
  number={2},
  pages={7--es},
  year={2007},
  publisher={ACM New York, NY, USA}
}

@article{jain2024scarce,
  title={Scarce resource allocations that rely on machine learning should be randomized},
  author={Jain, Shomik and Creel, Kathleen and Wilson, Ashia},
  journal={arXiv preprint arXiv:2404.08592},
  year={2024}
}

@article{de2024towards,
  title={Towards human-ai complementarity with predictions sets},
  author={De Toni, Giovanni and Okati, Nastaran and Thejaswi, Suhas and Straitouri, Eleni and Gomez-Rodriguez, Manuel},
  journal={arXiv preprint arXiv:2405.17544},
  year={2024}
}

@article{10.1145/2892565,
author = {Caragiannis, Ioannis and Procaccia, Ariel D. and Shah, Nisarg},
title = {When Do Noisy Votes Reveal the Truth?},
year = {2016},
issue_date = {June 2016},
publisher = {Association for Computing Machinery},
address = {New York, NY, USA},
volume = {4},
number = {3},
issn = {2167-8375},
url = {https://doi.org/10.1145/2892565},
doi = {10.1145/2892565},
journal = {ACM Trans. Econ. Comput.},
month = mar,
articleno = {15},
numpages = {30}
}

@inproceedings{boehmer2022quantitative,
  title={A quantitative and qualitative analysis of the robustness of (real-world) election winners},
  author={Boehmer, Niclas and Bredereck, Robert and Faliszewski, Piotr and Niedermeier, Rolf},
  booktitle={Proceedings of the 2nd ACM Conference on Equity and Access in Algorithms, Mechanisms, and Optimization},
  pages={1--10},
  year={2022}
}

@inproceedings{boehmer2023properties,
  title={Properties of the mallows model depending on the number of alternatives: a warning for an experimentalist},
  author={Boehmer, Niclas and Faliszewski, Piotr and Kraiczy, Sonja},
  booktitle={International Conference on Machine Learning},
  pages={2689--2711},
  year={2023},
  organization={PMLR}
}

@article{awasthi2014learning,
  title={Learning mixtures of ranking models},
  author={Awasthi, Pranjal and Blum, Avrim and Sheffet, Or and Vijayaraghavan, Aravindan},
  journal={Advances in Neural Information Processing Systems},
  volume={27},
  year={2014}
}

@book{kim2015human,
  title={Human-computer interaction: fundamentals and practice},
  author={Kim, Gerard Jounghyun},
  year={2015},
  publisher={CRC press}
}

@book{mackenzie2024human,
  title={Human-computer interaction: An empirical research perspective},
  author={MacKenzie, I. Scott},
  year={2024},
  publisher={Elsevier}
}

@article{collina2024tractable,
  title={Tractable Agreement Protocols},
  author={Collina, Natalie and Goel, Surbhi and Gupta, Varun and Roth, Aaron},
  journal={arXiv preprint arXiv:2411.19791},
  year={2024}
}

@article{corvelo2023human,
  title={Human-aligned calibration for ai-assisted decision making},
  author={Corvelo Benz, Nina and Rodriguez, Manuel},
  journal={Advances in Neural Information Processing Systems},
  volume={36},
  pages={14609--14636},
  year={2023}
}

@inproceedings{liu2018efficiently,
  title={Efficiently learning mixtures of mallows models},
  author={Liu, Allen and Moitra, Ankur},
  booktitle={2018 IEEE 59th Annual Symposium on Foundations of Computer Science (FOCS)},
  pages={627--638},
  year={2018},
  organization={IEEE}
}

@article{bradley1952rank,
  title={Rank analysis of incomplete block designs: I. The method of paired comparisons},
  author={Bradley, Ralph Allan and Terry, Milton E},
  journal={Biometrika},
  volume={39},
  number={3/4},
  pages={324--345},
  year={1952},
  publisher={JSTOR}
}

@book{lazar2017research,
  title={Research methods in human-computer interaction},
  author={Lazar, Jonathan and Feng, Jinjuan Heidi and Hochheiser, Harry},
  year={2017},
  publisher={Morgan Kaufmann}
}

@inproceedings{raman2014methods,
  title={Methods for ordinal peer grading},
  author={Raman, Karthik and Joachims, Thorsten},
  booktitle={Proceedings of the 20th ACM SIGKDD international conference on Knowledge discovery and data mining},
  pages={1037--1046},
  year={2014}
}

@book{preece1994human,
  title={Human-computer interaction},
  author={Preece, Jenny and Rogers, Yvonne and Sharp, Helen and Benyon, David and Holland, Simon and Carey, Tom},
  year={1994},
  publisher={Addison-Wesley Longman Ltd.}
}

@inproceedings{chan2019assistive,
 author = {Chan, Lawrence and Hadfield-Menell, Dylan and Srinivasa, Siddhartha and Dragan, Anca},
 booktitle = {2019 14th ACM/IEEE International Conference on Human-Robot Interaction (HRI)},
 organization = {IEEE},
 pages = {354--363},
 title = {The assistive multi-armed bandit},
 year = {2019}
}

@misc{straitouri2022provably,
    title={Provably Improving Expert Predictions with Conformal Prediction},
    author={Eleni Straitouri and Lequn Wang and Nastaran Okati and Manuel Gomez Rodriguez},
    year={2022},
    eprint={2201.12006},
    archivePrefix={arXiv},
    primaryClass={cs.LG}
}

@inproceedings{li2024decoding,
  title={Decoding ai’s nudge: A unified framework to predict human behavior in ai-assisted decision making},
  author={Li, Zhuoyan and Lu, Zhuoran and Yin, Ming},
  booktitle={Proceedings of the AAAI Conference on Artificial Intelligence},
  volume={38},
  number={9},
  pages={10083--10091},
  year={2024}
}

@article{ibrahim2025measuring,
  title={Measuring and mitigating overreliance is necessary for building human-compatible AI},
  author={Ibrahim, Lujain and Collins, Katherine M and Kim, Sunnie SY and Reuel, Anka and Lamparth, Max and Feng, Kevin and Ahmad, Lama and Soni, Prajna and Kattan, Alia El and Stein, Merlin and others},
  journal={arXiv preprint arXiv:2509.08010},
  year={2025}
}

@inproceedings{cowgill2020algorithmic,
  title={Algorithmic social engineering},
  author={Cowgill, Bo and Stevenson, Megan T},
  booktitle={AEA Papers and Proceedings},
  volume={110},
  pages={96--100},
  year={2020}
}

@article{rastogi2022unifying,
  title={A Unifying Framework for Combining Complementary Strengths of Humans and ML toward Better Predictive Decision-Making},
  author={Rastogi, Charvi and Leqi, Liu and Holstein, Kenneth and Heidari, Hoda},
  journal={arXiv preprint arXiv:2204.10806},
  year={2022}
}

@article{steyvers2022bayesian,
  title={Bayesian modeling of human--AI complementarity},
  author={Steyvers, Mark and Tejeda, Heliodoro and Kerrigan, Gavin and Smyth, Padhraic},
  journal={Proceedings of the National Academy of Sciences},
  volume={119},
  number={11},
  pages={e2111547119},
  year={2022},
  publisher={National Acad Sciences}
}

@inproceedings{madras2017predict,
 author = {Madras, David and Pitassi, Toni and Zemel, Richard},
 booktitle = {Advances in Neural Information Processing Systems},
 editor = {S. Bengio and H. Wallach and H. Larochelle and K. Grauman and N. Cesa-Bianchi and R. Garnett},
 pages = {},
 publisher = {Curran Associates, Inc.},
 title = {Predict Responsibly: Improving Fairness and Accuracy by Learning to Defer},
 url = {https://proceedings.neurips.cc/paper/2018/file/09d37c08f7b129e96277388757530c72-Paper.pdf},
 volume = {31},
 year = {2018}
}

@inproceedings{agarwal2023online,
  title={Online Recommendations for Agents with Discounted Adaptive Preferences},
  author={Brown, William and Agarwal, Arpit},
  booktitle={International Conference on Algorithmic Learning Theory},
  pages={244--281},
  year={2024},
  organization={PMLR}
}

@article{mallows1957non,
  title={Non-null ranking models. I},
  author={Mallows, Colin L},
  journal={Biometrika},
  volume={44},
  number={1/2},
  pages={114--130},
  year={1957},
  publisher={JSTOR}
}

@article{thurstone1927law,
  title={A law of comparative judgment.},
  author={Thurstone, Louis L},
  journal={Psychological review},
  volume={34},
  number={4},
  pages={273},
  year={1927},
  publisher={Psychological Review Company}
}

@article{corvelo2025human,
  title={Human-Alignment Influences the Utility of AI-assisted Decision Making},
  author={Corvelo Benz, Nina Laura and Gomez Rodriguez, Manuel},
  journal={arXiv preprint arXiv:2501.14035},
  year={2025}
}

@misc{babbar2022utility,
      title={On the Utility of Prediction Sets in Human-AI Teams}, 
      author={Varun Babbar and Umang Bhatt and Adrian Weller},
      year={2022},
      eprint={2205.01411},
      archivePrefix={arXiv},
      primaryClass={cs.AI}
}

@misc{straitouri2023designing,
      title={Designing Decision Support Systems Using Counterfactual Prediction Sets}, 
      author={Eleni Straitouri and Manuel Gomez Rodriguez},
      year={2023},
      eprint={2306.03928},
      archivePrefix={arXiv},
      primaryClass={cs.LG}
}

@article{angelopoulos2020uncertainty,
  title={Uncertainty sets for image classifiers using conformal prediction},
  author={Angelopoulos, Anastasios and Bates, Stephen and Malik, Jitendra and Jordan, Michael I},
  journal={arXiv preprint arXiv:2009.14193},
  year={2020}
}

@inproceedings{wang2022improving,
  title={Improving screening processes via calibrated subset selection},
  author={Wang, Lequn and Joachims, Thorsten and Rodriguez, Manuel Gomez},
  booktitle={International Conference on Machine Learning},
  pages={22702--22726},
  year={2022},
  organization={PMLR}
}

@inproceedings{donahue2022human,
  title={Human-algorithm collaboration: Achieving complementarity and avoiding unfairness},
  author={Donahue, Kate and Chouldechova, Alexandra and Kenthapadi, Krishnaram},
  booktitle={Proceedings of the 2022 ACM Conference on Fairness, Accountability, and Transparency},
  pages={1639--1656},
  year={2022}
}

@Article{benCSCW,
  author    = {Green, Ben and Chen, Yiling},
  journal   = {Proceedings of the ACM on Human-Computer Interaction},
  title     = {The principles and limits of algorithm-in-the-loop decision making},
  year      = {2019},
  publisher = {ACM New York, NY, USA},
}

@article{alur2024human,
  title={Human expertise in algorithmic prediction},
  author={Alur, Rohan and Raghavan, Manish and Shah, Devavrat},
  journal={Advances in Neural Information Processing Systems},
  volume={37},
  pages={138088--138129},
  year={2024}
}

@article{alur2023auditing,
  title={Auditing for human expertise},
  author={Alur, Rohan and Laine, Loren and Li, Darrick and Raghavan, Manish and Shah, Devavrat and Shung, Dennis},
  journal={Advances in Neural Information Processing Systems},
  volume={36},
  pages={79439--79468},
  year={2023}
}

@inproceedings{chen2025missing,
  title={Missing Pieces: How Do Designs that Expose Uncertainty Longitudinally Impact Trust in AI Decision Aids? An In Situ Study of Gig Drivers},
  author={Chen, Rex and Wang, Ruiyi and Sadeh, Norman and Fang, Fei},
  booktitle={Proceedings of the 2025 ACM Conference on Fairness, Accountability, and Transparency},
  pages={790--816},
  year={2025}
}

@inproceedings{greenwood2024designing,
  title={Designing Algorithmic Delegates},
  author={Greenwood, Sophie and Levy, Karen and Barocas, Solon and Kleinberg, Jon and Heidari, Hoda},
  booktitle={NeurIPS 2024 Workshop on Behavioral Machine Learning},
year = {2024}
}

@article{peng2024no,
  title={A No Free Lunch Theorem for Human-AI Collaboration},
  author={Peng, Kenny and Garg, Nikhil and Kleinberg, Jon},
  journal={arXiv preprint arXiv:2411.15230},
  year={2024}
}

@article{vaccaro2024combinations,
  title={When combinations of humans and AI are useful: A systematic review and meta-analysis},
  author={Vaccaro, Michelle and Almaatouq, Abdullah and Malone, Thomas},
  journal={arXiv preprint arXiv:2405.06087},
  year={2024}
}

@article{gomez2025human,
  title={Human-AI collaboration is not very collaborative yet: a taxonomy of interaction patterns in AI-assisted decision making from a systematic review},
  author={Gomez, Catalina and Cho, Sue Min and Ke, Shichang and Huang, Chien-Ming and Unberath, Mathias},
  journal={Frontiers in Computer Science},
  volume={6},
  pages={1521066},
  year={2025},
  publisher={Frontiers Media SA}
}

@inproceedings{bansal2021most,
  title={Is the most accurate ai the best teammate? optimizing ai for teamwork},
  author={Bansal, Gagan and Nushi, Besmira and Kamar, Ece and Horvitz, Eric and Weld, Daniel S},
  booktitle={Proceedings of the AAAI Conference on Artificial Intelligence},
  volume={35},
  pages={11405--11414},
  year={2021}
}

@article{Agarwal2022DiversifiedRF,
  title={Diversified recommendations for agents with adaptive preferences},
  author={Brown, William and Agarwal, Arpit},
  journal={Advances in Neural Information Processing Systems},
  volume={35},
  pages={26066--26077},
  year={2022}
}

@String{Computing = "Computing" }

@String{Computer = "{IEEE} Computer" }

@article{ge2024axioms,
  title={Axioms for ai alignment from human feedback},
  author={Ge, Luise and Halpern, Daniel and Micha, Evi and Procaccia, Ariel D and Shapira, Itai and Vorobeychik, Yevgeniy and Wu, Junlin},
  journal={arXiv preprint arXiv:2405.14758},
  year={2024}
}

@article{a2024policy,
  title={Policy aggregation},
  author={A Alamdari, Parand and Ebadian, Soroush and Procaccia, Ariel D},
  journal={Advances in Neural Information Processing Systems},
  volume={37},
  pages={68308--68329},
  year={2024}
}

@article{golz2025distortion,
  title={Distortion of AI Alignment: Does Preference Optimization Optimize for Preferences?},
  author={G{\"o}lz, Paul and Haghtalab, Nika and Yang, Kunhe},
  journal={arXiv preprint arXiv:2505.23749},
  year={2025}
}

@article{shirali2025direct,
  title={Direct Alignment with Heterogeneous Preferences},
  author={Shirali, Ali and Nasr-Esfahany, Arash and Alomar, Abdullah and Mirtaheri, Parsa and Abebe, Rediet and Procaccia, Ariel},
  journal={arXiv preprint arXiv:2502.16320},
  year={2025}
}

@article{conitzer2024social,
  title={Social choice for ai alignment: Dealing with diverse human feedback},
  author={Conitzer, Vincent and Freedman, Rachel and Heitzig, Jobst and Holliday, Wesley H and Jacobs, Bob M and Lambert, Nathan and Moss{\'e}, Milan and Pacuit, Eric and Russell, Stuart and Schoelkopf, Hailey and others},
  journal={arXiv preprint arXiv:2404.10271},
  year={2024}
}

@article{dai2024mapping,
  title={Mapping social choice theory to RLHF},
  author={Dai, Jessica and Fleisig, Eve},
  journal={arXiv preprint arXiv:2404.13038},
  year={2024}
}

@article{siththaranjan2023distributional,
  title={Distributional preference learning: Understanding and accounting for hidden context in rlhf},
  author={Siththaranjan, Anand and Laidlaw, Cassidy and Hadfield-Menell, Dylan},
  journal={arXiv preprint arXiv:2312.08358},
  year={2023}
}

@article{chen2024pal,
  title={Pal: Pluralistic alignment framework for learning from heterogeneous preferences},
  author={Chen, Daiwei and Chen, Yi and Rege, Aniket and Vinayak, Ramya Korlakai},
  journal={arXiv preprint arXiv:2406.08469},
  year={2024}
}

@article{pardeshi2024learning,
  title={Learning social welfare functions},
  author={Pardeshi, Kanad and Shapira, Itai and Procaccia, Ariel D and Singh, Aarti},
  journal={Advances in Neural Information Processing Systems},
  volume={37},
  pages={41733--41766},
  year={2024}
}

@inproceedings{singh2018fairness,
  title={Fairness of exposure in rankings},
  author={Singh, Ashudeep and Joachims, Thorsten},
  booktitle={Proceedings of the 24th ACM SIGKDD international conference on knowledge discovery \& data mining},
  pages={2219--2228},
  year={2018}
}

@article{collina2025emergent,
  title={Emergent Alignment via Competition},
  author={Collina, Natalie and Goel, Surbhi and Roth, Aaron and Ryu, Emily and Shi, Mirah},
  journal={arXiv preprint arXiv:2509.15090},
  year={2025}
}

\newpage
\appendix
\section{Properties of Noisy Permutation Model}\label{app:propertyperm}
\subsection{Basic Properties}
Let $\mathfrak{S}(M)$ be the set of all rankings of $M$.
Denote by $x_i \ogt{\pi} S$ for a set of items $S$ if $x_i$ is before any item $x_j\in S\setminus \{x_i\}$. 
A swap $(x_i x_j)$ is \emph{valid} with respect to $\pi$ and a ground-truth ranking $\pi^*$ if $\pi$ and $\pi^*$ differ on the relative ranking of $x_i$ and $x_j$. 
By the definition of inversion-monotonicity, if $(x_i x_j)$ is a valid swap with respect to any given ranking $\pi$ and a ground-truth ranking $\pi^*$,  then $\pi\circ(x_i, x_j)$ has a higher probability than $\pi$. 
Moreover, according to \cite[Lemma 1]{donahue2024listsbetteronebenefits}, the number of inversions reduces by at least one after the swap.

\begin{lemma}\label{lem:mono_presenting_items}
Let $\D$ be an inversion monotonic distribution with ground-truth ranking $\pi^*$.
The random permutation $\pi$ is sampled from $\D$.
For any subset $\mS$ of $k$ items, the probability of $\mS$ being the first $k$ items of $\pi$ decreases if substituting $x_i\in \mS$ with another item $x_j\notin \mS$ with $x_i \succ_{\pi^*} x_j$.
\end{lemma}
\begin{proof}
Let $\mS' = \mS\setminus \{x_i\}\cup \{x_j\}$.
Define two set of permutations $\mathfrak{S}_{\mS}$ and $\mathfrak{S}_{\mS'}$ as the set of permutations that place $\mS$ and $\mS'$ at the first $k$ locations, respectively.
Then we define a bijective mapping from a permutation $\pi \in \mathfrak{S}_{\mS}$ to a permutation $\pi' \in \mathfrak{S}_{\mS'}$ by swapping the items $x_i$ and $x_j$.
By inverse-monotonicity, the probability of $\pi'$ is (weakly) higher than that of $\pi$.
Therefore, we can conclude that the probability of $\mS'$ being the first $k$ of $\pi$ is greater than that of $\mS$.
\end{proof}

\begin{restatable}{lemma}{LemCompMono}
\label{lem:CompMono}
Let $\D$ be an inversion monotonic distribution with ground-truth ranking $\pi^*$.
A permutation $\pi$ is drawn from the distribution $\D$.
Let $\mS$ be a subset of $k$ items such that $x_\ell \in \mS$, $x_r \notin \mS$, and $x_\ell \succ_{\pi^*} x_r$. Define $\mS' = (\mS \setminus {x_\ell}) \cup {x_r}$. Then, for any item $x_i \in \mS\setminus \{x_\ell\}$, the probability that $x_i$ is ranked first among $\mS$ in $\pi$ is less than the probability that it is ranked first among $\mS'$ in $\pi$: 
$$
\myP[x_i \succ_\pi \mS] \le \myP[x_i \succ_\pi \mS']\,.
$$
The inequality becomes tight only when $x_r$ and $x_\ell$ are totally indistinguishable under $\D$.
\end{restatable}
\begin{proof}
Let $\mathfrak{S}_R = \{\pi: \text{$x_i$ is before $\mS'$ in $\pi$}\}$, and $\mathfrak{S}_L = \{\pi: \text{$x_i$ is before $\mS$ in $\pi$}\}$.
To show the inequality, it suffices to prove $\myP[\pi \in \mathfrak{S}_R \setminus \mathfrak{S}_L] >
\myP[\pi \in \mathfrak{S}_L \setminus \mathfrak{S}_R]$.
By definitions of $\mathfrak{S}_L$ and $\mathfrak{S}_R$, $\mathfrak{S}_R \setminus \mathfrak{S}_L$ contains the permutations where $x_\ell$ is placed before $x_i$ and $x_i$ is placed before other items of $\mS'$ and $x_r$.
Similarly, $\mathfrak{S}_L \setminus \mathfrak{S}_R$ consists of the permutations where $x_r$ is placed before $x_i$ and $x_i$ is placed before other items of $\mS$ (including $x_\ell$).
Define a mapping from $\mathfrak{S}_L \setminus \mathfrak{S}_R$ to $\mathfrak{S}_R\setminus \mathfrak{S}_R$ by swapping $x_\ell$ and $x_r$, which is a valid mapping by definition.
As the distribution is inversion monotonic, the probability of the new permutation is greater than the original one, which implies that $\myP[\pi \in \mathfrak{S}_R \setminus \mathfrak{S}_L] > 
\myP[\pi \in \mathfrak{S}_L \setminus \mathfrak{S}_R]$ and concludes the lemma.
\end{proof}

\subsection{Properties of Mallows Model}
\label{app:property_of_mallows}

In \Cref{tab:mallows_properties}, we provide an overview of the properties of the Mallows model used in our proofs, where $\pi\sim \D(\pi^*, \phi)$ and $Z_m(\phi) = \sum_{i=1}^m \exp(-\phi(i-1))$.
Without loss of generality, we assume $\pi^* = (x_1, \ldots, x_m)$ in this subsection.

\begin{table*}[ht]
\centering
\begin{tabular}{l||l|l}
\hline
{\bf Event} & {\bf Probability} & {\bf Source}  \\
\hline
$x_i$ ranked the first in $\pi$  &  $\frac{\exp(-\phi\cdot (i-1))}{Z_m(\phi)}$ &  \citep{awasthi2014learning}\\ 
$x_1$ ranked at the $i$-th location  in $\pi$ & $\frac{\exp(-\phi\cdot (i-1))}{Z_m(\phi)}$ & \citep{awasthi2014learning} \\ 
$x_1$ ranked at the first $k$ locations in $\pi$ &  $ \frac{Z_k(\phi)}{Z_m(\phi)}$ & \citep{awasthi2014learning} \\ 
$x_i$ ranked before $x_j$ with $i< j$ and $k=j-i+1$ in $\pi$ &  $\frac{k}{1-\exp(-\phi\cdot k)} - \frac{k-1}{1-\exp(-\phi\cdot (k-1))}$ &  \citep{boehmer2023properties} \\
a set of items $S$ being the first $k$ items of $\pi$ & closed-formed & \Cref{lem:MallowsFirstk} \\
the first item of $S$ in $\pi$ is the best among $S$ in $\pi^*$ & $\ge \frac{1}{Z_k(\phi)}$ with $k=\abs{S}$ & \Cref{lem:LowerBoundForComparison}\\
 the first item of $S$ in $\pi$ is within the top $s$-th among $S$ in $\pi^*$ &  $\ge \frac{Z_s(\phi)}{Z_k(\phi)}$ with $k=\abs{S}$ & \Cref{lem:general_lower_bound_for_comparison}\\
\hline
\end{tabular}
\caption{Properties of the Mallows model.}
\label{tab:mallows_properties}
\end{table*}

\addtocontents{toc}{\protect\setcounter{tocdepth}{1}}
\subsubsection{Item Location}
\label{app:item_loc}
\addtocontents{toc}{\protect\setcounter{tocdepth}{2}}

We first show that, for every set of $\mS$ with $k$ items, the probability of $\mS$ being the first $k$ items of $\ar$ can have a closed-form expression in the Mallows model using the insertion-based sampling.
Consider a fixed order $\pi_S= (x_{j(1)}, \ldots, x_{j(k)})$ of the $k$ items.
We insert the $k$ items in order $\pi_S$.
At the time of inserting the $i$-th one, the probability of placing $x_{j(i)}$ at the $i$-th location conditioned on the previous insertion, depends on the relative location of $x_{j(1)}$ among $\mS$ and the inserted items.
Consequently, it implies the probability of the prefix is solely determined by the number of inversions in $\pi_S$.
By summing it up for all $\pi_S$, we get a closed-form probability of $\mS$ being the first $k$, which only depends on the sum of indices of items of $\mS$.
The formal theorem statement and proof are presented as follows:

\begin{restatable}{lemma}{LemMallowsFirstk}
\label{lem:MallowsFirstk}
In a Mallows distribution $\pi\sim \mathcal{D}(\pi^*,\phi)$, given a subset $S$ of items of size $k$, the probability of set $S$ of being the first $k$ items of $\pi$ is given by 
$$
\myP[\ars{k} = S] = \exp\left(-\phi\cdot \left(\sum\nolimits_{x_j\in S}j - \frac{k(k+1)}{2}\right)\right)\cdot \prod_{i=1}^{k} \frac{Z_i(\phi)}{Z_{m-i+1}(\phi)}\,.
$$
\end{restatable}
\begin{proof}
Fix an order $\pi_S$ of $S$.
Below, we give the probability of $\pi_S$ being the prefix of $\pi$.
Denote the $k$ items by $x_{j(1)}, \ldots, x_{j(k)}$ according order $\pi_S$.
We insert the $k$ items in order.
At time of inserting the $i$-th item $x_{j(i)}$, consider the probability of locating $x_{j(i)}$ at the $i$-th place in $\pi$ conditioned on $x_{j(1)}, \ldots, x_{j(i-1)}$ being located as the first $i-1$ items of $\pi$.
Observe that the number of inserted items that are placed before $x_{j(i)}$ in the ground truth ranking $\pi^*$ is equal to $\abs{\{i': j(i') < j(i) \wedge i' < i\}}$.
Denote that number by $A(i)$.
The current item $x_{j(i)}$ is the $(j(i) - A(i))$-th item among the set of remaining items $M \setminus \{x_{j(t)}\}_{t=1}^{i-1}$.
Then, according to the property of the Mallows model, the probability of $x_{j(i)}$ being placed in $i$-th place in $\pi$ conditioned on the past insertion is given by 
\begin{align}\label{eqn:mallows_cond_prob}
\myP\left[\pi(i) = x_{j(i)} \mid \wedge_{t=1}^{i-1} \pi(t) = x_{j(t)} \right] = \frac{1}{Z_{m-i+1}(\phi)}\cdot \exp\left(-\phi\cdot (j(i) - A(i) - 1)\right)\,.
\end{align}
Based on it, the probability of $\pi_S$ being the prefix of $\pi$ is equal to the product of these conditional probabilities, which is given by
\begin{align*}
\myP\left[\wedge_{i=1}^k \pi(i) = x_{j(i)}\right] & = \prod_{i=1}^k \Pr\left[\pi(i) = x_{j(i)} \mid \wedge_{t=1}^{i-1} \pi(t) = x_{j(t)} \right] \tag{By the chain rule}\\ 
& = \prod_{i=1}^{k} \frac{\exp\left(-\phi\cdot (j(i) - A(i) - 1)\right)}{Z_{m-i+1}(\phi)} \tag{By \Cref{eqn:mallows_cond_prob}}
\end{align*}
By direct calculation, we can see that
\begin{align*}
\myP\left[\wedge_{i=1}^k \pi(i) = x_{j(i)}\right] & = \exp\left(-\phi\cdot \sum_{i=1}^k(j(i) - i + i - 1 - A(i))\right) \cdot \prod_{i=1}^k \frac{1}{Z_{m-i+1}(\phi)}\,.
\end{align*}
Notice that the term $i - 1 -A(i)$ is essentially the number of items that are placed before $x_{j(i)}$ in $\pi$ but ranked after $x_{j(i)}$ in the ground-truth ranking $\pi^*$.
Hence, each of them forms an inversion with $x_{j(i)}$.
Thus, the sum of $i - 1 -A(i)$ over $i$ equals to the number of inversions of $\pi_S$.
Therefore, we have
\begin{align}\label{eqn:prob_specific_perm}
\myP\left[\wedge_{i=1}^k \pi(i) = x_{j(i)}\right] = \exp\Big(-\phi\cdot \Big(\sum_{x_j\in S}j - \frac{k(k+1)}2\Big)\Big)\cdot \exp(-\phi\cdot \mathrm{inv}(\pi_S))\cdot \prod_{i=1}^k \frac{1}{Z_{m-i+1}(\phi)}\,.
\end{align}
By summing the above equation over all the permutations of items in $S$, we can obtain that
\begin{align*}
\myP[\ars{k} = S] = \sum_{\pi_S\in \mathfrak{S}(S)}\myP\left[\wedge_{i=1}^k \pi(i) = x_{j(i)}\right].
\end{align*}
Then applying \Cref{eqn:prob_specific_perm}, since the sum of $\exp(-\phi\cdot \mathrm{inv}(\pi_S))$ over all permutations of $S$ is equal to $\prod_{i=1}^k Z_i(\phi)$~\citep{mallows1957non}, we then have 
\begin{equation*}
\myP\left[\pi(:k) = S\right] = \exp\left(-\phi\cdot \left(\sum_{x_j\in S}j - \frac{k(k+1)}2\right)\right)\cdot  \frac{\prod_{i=1}^k Z_i(\phi)}{\prod_{i=1}^k Z_{m-i+1}(\phi)}\,.
\qedhere
\end{equation*}
\end{proof}   

\addtocontents{toc}{\protect\setcounter{tocdepth}{1}}
\subsubsection{Itemwise Comparison}
\label{sec:pairwise_comp}
\addtocontents{toc}{\protect\setcounter{tocdepth}{2}}

\begin{lemma}[Proposition 3.8 of \citep{boehmer2023properties}]\label{lem:preferonetoother}
In a Mallows distribution $\pi\sim \mathcal{D}(\pi^*, \phi)$,  for any two distinct items $x_i, x_j\in M$ with $i< j$ and $k=j-i+1$, the probability that $x_i$ is before $x_j$ in $\pi$ is given by 
$$
\Pr\left[x_i \succ_\pi x_j\right] = \frac{k}{1-\exp(-\phi\cdot k)} - \frac{k-1}{1-\exp(-\phi\cdot (k-1))}\,.
$$
\end{lemma}

\begin{lemma}[Monotonicity of Pairwise Comparison]
\label{lem:mono_pairwise}
In a Mallows distribution $\pi\sim \mathcal{D}(\pi^*, \phi)$, the probability of the pairwise comparison between items $x_i$ and $x_j$ with $j> i$ increases with the index gap $j-i$.
Particularly, $\myP[x_i \succ_{\pi} x_j] \le \myP[x_{i'} \succ_{\pi} x_{j'}]$ if $j'-i'\ge j-i$.
\end{lemma}
\begin{proof}
We first prove it holds for the case of $i=i'$.
Hence, by assumption, $j'\ge j$.
For any ranking $\pi$ that ranks $x_i$ before $x_{j}$ while ranking $x_i$ after $x_{j'}$, we map it to a new one by swapping $x_{j}$ and $x_{j'}$.
By \cite[Lemma 1]{donahue2024listsbetteronebenefits}, the number of inversions reduces by at least one, and the probability of the new ranking is weakly larger than the original one.
Therefore, it implies that the probability of $x_i$ being before $x_j$ is smaller.
By \Cref{lem:preferonetoother}, the probability of pairwise comparison only depends on the difference of the two indices.
Therefore, the lemma holds for every $i, j, i'$, and $j'$.
\end{proof}

\begin{lemma}\label{lem:set_constant}
For any $n$ and a constant $\epsilon < 1/2$, there exists $\phi > 0$ and $m$ such that in Mallows distribution $\pi \sim \D(\pi^*, \phi)$ with ground-truth ranking $\pi^* = (x_1, \ldots, x_m)$, for any two items $x_i$ and $x_j$, the following holds
\begin{align*}
\myP[x_i \succ_\pi x_j] \in \begin{cases}
(0, \frac12 + \epsilon], & \text{if } j- i +1 \le n \\
[1 - \epsilon, 1), & \text{if } j- i +1 \ge m - n
\end{cases}
\end{align*}
\end{lemma}
\begin{proof}
When $j - i +1 \le n$, we first show an upper bound of the probability of $x_i$ being before $x_j$.
Similarly, we define a bijective mapping from the rankings where $x_i$ is before $x_j$ to rankings where $x_i$ is after $x_i$ by swapping the location of $x_i$ and $x_j$.
Since $j - 1 + 1\le n$, the number of inversions will increase by at most $n$.
Hence, $\myP[x_i \succ_\pi x_j] \le \exp(\phi\cdot n)$, which immediately implies that, $\myP[x_i \succ_\pi x_j] \le \frac{\exp(\phi\cdot n)}{1+ \exp(\phi\cdot n)}$.
Thus, by direct calculation, it suffices to set $\phi = \frac1n\cdot \log(\frac{1+2\epsilon}{1 - 2\epsilon})$ for the first condition.

Next, we find the number $m$ satisfying the second condition. 
By \Cref{lem:mono_pairwise}, for any fixed $\phi$, the probability increases as $k$.
Hence, it suffices to show it holds for $m-n$.
Recall that, by \Cref{lem:preferonetoother},
$$
\myP[x_i \succ_\pi x_j] = \frac{k}{1-\exp(-\phi\cdot k)} - \frac{k-1}{1-\exp(-\phi\cdot (k-1))}\,.
$$
Direct calculation gives that
\begin{align*}
\myP[x_i \succ_\pi x_j] 
& = \frac{1}{1-\exp(-\phi\cdot (k-1))}\cdot \left(1 - \frac{\exp(-\phi\cdot (k-1)) - \exp(-\phi\cdot k)}{1- \exp(-\phi\cdot k)}k\right) \\ 
& \ge  \frac{1}{1-\exp(-\phi\cdot (k-1))}\cdot \left(1- \exp(-\phi\cdot (k-1))\cdot k \right)\,.
\end{align*}
To make sure the probability is greater than $1-\epsilon$, it suffices to set $\exp(-\phi\cdot (k-1)) \le \epsilon/k$, which can be satisfied when $k/\log k \ge 1/\phi$.
As $\log k \le \sqrt{k}$, it suffices to set $m - n \ge k \ge 1/\phi^2$.
Therefore, we set $\phi = \frac1n\cdot \log(\frac{1+2\epsilon}{1 - 2\epsilon})$ and $m = 1/\phi^2 + n$.
\end{proof}

\begin{restatable}{lemma}{LemLowerBoundForComparison}
\label{lem:LowerBoundForComparison}
In a mallows distribution $\pi\sim \D(\pi^*, \phi)$, given a subset of items $\mS =\{x_{i(1)}, \ldots, x_{i(k)}\}$ with $i(1) < \ldots < i(k)$, the probability of $x_{i(1)}$ being placed before all other items of $\mS$ in $\pi$ satisfies
$$
\myP_{\pi \sim \mathcal{D}(\pi^*, \phi)}\left[x_{i(1)} \succ_{\pi} \mS \right] \ge \frac{1}{Z_k(\phi)}\,.
$$
\end{restatable}
\begin{proof}
We partition the set of permutations into $k$ sets based on the relative position of $x_{i(1)}$.
Let $\mathfrak{S}_t$ be the set of permutations where $x_{i(1)}$ is placed at the $t$-th relative position among all items of $\mS$.
For any $\pi \in \mathfrak{S}_{t+1}$ with $t\le k-1$, we map it to a permutation $\pi \in \mathfrak{S}_t$ by swapping $x_{i(1)}$ and the item that is placed at the $t$-th position among all items of $\mS$.
It is a valid swapping with respect to $\pi$ and $\pi^*$ by definition.
By \cite[Lemma 1]{donahue2024listsbetteronebenefits}, the probability of $\pi'$ is at least $\exp(\phi)$ times the probability of $\pi$.
By summing the inequality up over all permutations of $\mathfrak{S}_{t+1}$, we can conclude that $\myP[\pi \in \mathfrak{S}_t] \ge \exp(\phi)\cdot \myP[\pi \in \mathfrak{S}_{t+1}]$.
Hence, by multiplying these inequalities, we can have
$$
\myP[\pi \in \mathfrak{S}_1] \ge \exp(\phi\cdot (t-1))\cdot \myP[\pi \in \mathfrak{S}_t]\,.
$$ 
Therefore, by summing it over all $t$, we have
\begin{equation*}
\myP[\pi\in \mathfrak{S}_1] \ge \frac{1}{\sum_{t=1}^k \exp(\phi\cdot (t-1))} \cdot \sum_{t=1}^k \myP[\pi\in \mathfrak{S}_t] = \frac{1}{Z_k(\phi)}.
\qedhere    
\end{equation*}
\end{proof}

\begin{restatable}{lemma}{LemGeneralLowerBoundForComparison}
\label{lem:general_lower_bound_for_comparison}
In a mallows distribution $\pi\sim \mathcal{D}(\pi^*, \phi)$ with $\pi^* = (x_1, \ldots, x_m)$, given a subset of items $\mS=\{x_{i(1)}, \ldots, x_{i(k)}\}$ with $i(1) < \ldots < i(k)$, for any $s$ such that $1\le s\le k$, the probability of one of $x_{i(1)}, \ldots, x_{i(s)}$ being the first among all items of $\mS$ is at least
$$
\sum_{j=1}^s \myP \left[x_{i(j)} \succ_{\pi} \mS\right] \ge \frac{Z_s(\phi)}{Z_k(\phi)}\,.
$$
\end{restatable}
\begin{proof}
Similar to the above proof, let $\mathfrak{S}_j$ be the set of permutations where $x_{i(j)}$ is the first one among all the items in $\mS$.
Then the left-hand side can be rewritten as 
$\sum_{j=1}^s \myP[\pi \in \mathfrak{S}_j]$. 
Denote each term by $f(j)$ for $j = 1, \ldots, s$.
Using the argument of valid-swapping-based mapping, we have $f(i) \ge \exp(\phi)\cdot f(i+1)$.
Meanwhile, as $\sum_{i=1}^k f(i) = 1$,  we have 
\begin{align*}
& \sum_{j=1}^s f(j) -  \frac{Z_s(\phi)}{Z_k(\phi)} = \sum_{j=1}^s f(j) -  \frac{Z_s(\phi)}{Z_k(\phi)} \cdot \sum_{j=1}^k f(j) \\ 
& \propto \sum_{j=1}^s f(j)\cdot (Z_k(\phi) - Z_s(\phi)) -  \sum_{j=s+1}^k f(j)\cdot Z_s(\phi) \tag{By multiplying $Z_k(\phi)$}\\
& = \sum_{j=1}^s f(j)\cdot \exp(-\phi\cdot s)\cdot Z_{k-s}(\phi) - \sum_{j=s+1}^k f(j)\cdot Z_s(\phi)\tag{$Z_k(\phi) - Z_s(\phi) = \frac{Z_{k-s}(\phi)}{\exp(\phi\cdot s)}$}
\end{align*}
Since $Z_{k-s}(\phi)$ can be unfolded as $Z_{k-s}(\phi) = \sum_{t=1}^{k-s} \exp(-\phi(t-1))$, by changing the order of summation, we can obtain that
\begin{align*}
\sum_{j=1}^s f(j) -  \frac{Z_s(\phi)}{Z_k(\phi)}  & \propto \sum_{t=1}^{k-s} \exp(-\phi\cdot (t- 1))\cdot \sum_{j=1}^s f(j)\cdot \exp(-\phi\cdot s)- \sum_{j=1}^k f(j)\cdot Z_s(\phi) \\
& =  \sum_{t=s+1}^{k} \exp(-\phi\cdot (t- 1))\cdot \sum_{j=1}^s f(j) - \sum_{j=s+1}^k f(j)\cdot Z_s(\phi)\\
& = \sum_{t=s+1}^{k}\sum_{j=1}^s \left(f(j)\cdot \exp(-\phi\cdot (t- 1))\right) -\sum_{j=s+1}^k f(j)\cdot Z_s(\phi)\\
& \ge \sum_{t=s+1}^{k}\sum_{j=1}^s f(t)\cdot \exp(-\phi\cdot (j-1)) - \sum_{j=s+1}^k f(j)\cdot Z_s(\phi) \tag{$f(i) \ge f(i+1)$} \\
& = \sum_{t=s+1}^{k}f(t)\cdot Z_s(\phi) - \sum_{j=s+1}^k f(j)\cdot Z_s(\phi) = 0,
\end{align*}
which concludes the proof.
\end{proof}

\section{Omitted Details of Section 4}
\label{app:omitted_proofs_of_sec4}

\subsection{Omitted Proofs of \Cref{thm:benefit_inv1}}

\LemProbOfiP*
\begin{proof}
Let $\ar_2^*$ be the ground-truth ranking of $A_2$, which swaps the ranking $x_i$ and $x_j$ compared to $A_1$'s ground-truth ranking $\ar_1^*$.
Let $\D_a^1$ and $\D_a^2$ be the distributions with ground-truth ranking $\ar_1^*$ and $\ar_2^*$ respectively.
In addition, denote by $x_C^1$ and $x_C^2$ the two items picked by the human working with Algorithm 1 and Algorithm 2, respectively.

\paragraph{Case I: Picking item other than $x_i$ and $x_j$.}
We first consider the change in the probability of picking some item other than $x_i$ or $x_j$.
Without loss of generality, denote by $x_r$ the considered item.
Denote by $x_C^1$ and $x_C^2$, respectively, the items picked by the human when collaborating with algorithm 1 and algorithm 2.
By the way of human-algorithm interaction, $x_r$ is chosen by a human only if it is included in the presented list $\mS$ (i.e., $\ars{k}= \mS$) and the human ranks it before any other item of $\mS$ (i.e., $x_r \succ_{\hr} \mS$).
Thus, the probabilities of $x_r$ being picked by the human are given by
\begin{align}
\myP[x_C^1 = x_r] &= \sum_{\mS: x_r \in \mS } \myP\left[\aris{1}{k}= \mS \right]\times \myP\left[x_r \succ_{\hr} \mS \right] \tag{Prob 1}, 
\label{prob:picking_xi_system1}\\
\myP[x_C^2 = x_r] &= \sum_{\mS: x_r \in \mS} \myP\left[\aris{2}{k} = \mS \right]\times \myP\left[x_r \succ_{\hr} \mS \right]
\label{prob:picking_xi_system2}\tag{Prob 2},
\end{align}
where $\pi_1\sim \D_a^1$, $\pi_2\sim \D_a^2$, and $\rho\sim \D_h$.

Next, we conduct a term-by-term comparison of \ref{prob:picking_xi_system1} and \ref{prob:picking_xi_system2} with respect to each $\mS$.
We start with two simple cases: $\mS$ contains both $x_i$ and $x_j$; $\mS$ contains neither $x_i$ nor $x_j$.
For either of the two cases, consider the permutation $\ar_1$ whose first $k$ items are $\mS$.
Swap $x_i$ and $x_j$ in $\ar_1$ and get another permutation $\ar_2$, which shares the same set of the first $k$ items.
As the two distributions only differ in the relative rankings of $x_i$ and $x_j$ in their ground-truth rankings, the probability of $\ar_1$ occurring in $\D_a^1$ is the same as $\ar_2$ occurring in $\D_a^2$.
Therefore, the two summations are equal in these terms.

Next, we compare the two probabilities by pairing the remaining summation terms corresponding to $\mS$ that contains only one of $x_i$ and $x_j$.
Denote by $\mS^i$ a set $\mS$ containing $x_i$ without $x_j$.
We pair $\mS^i$ with another subset by substituting item $x_i$ with $x_j$.
Denote the new set by $\mS^j$.
Next, we show that the following inequality holds for every constructed pair $(\mS^i, \mS^j)$,
\begin{align}\label{ineq:pair_sl_sr}
\sum_{\mS\in (\mS^i, \mS^j)} \myP\left[\aris{1}{k}= \mS \right] \myP_{\hr \sim \D_h}\left[x_r \succ_{\hr} \mS \right] \le &\sum_{\mS\in (\mS^i, \mS^j)} \myP\left[\aris{2}{k} = \mS \right]\myP_{\hr \sim \D_h}\left[x_r \succ_{\hr} \mS \right] \tag{InEq (1)}\,.
\end{align}
As $x_i$ is better than $x_j$ in human's ground-truth ranking $\hr^*$, it is easier for the human to rank $x_r$ before $x_j$, as described in \Cref{lem:CompMono}.
Formally,
$$
\myP[x_r \succ_{\hr} \mS^j] \ge  \myP[x_r \succ_{\hr} \mS^i]\,.     
$$
Also, since the two distributions $\D_a^1$ and $\D_a^2$ are isomorphic, differing only by a swap between $x_i$ and $x_j$ in ground-truth rankings, by relabeling the two items, we have 
\begin{align*}
&\myP_{\ar_1\sim \D^1_a} \left[\aris{1}{k} = \mS^i\right]  = 
\myP_{\ar_2\sim \D^2_a} \left[\aris{2}{k} = \mS^j\right],\\
&\myP_{\ar_1\sim \D^1_a} \left[\aris{1}{k} = \mS^j\right]  = 
\myP_{\ar_2\sim \D^2_a} \left[\aris{2}{k} = \mS^i\right]\,.
\end{align*}
Meanwhile, by \Cref{lem:mono_presenting_items}, as $x_i$ is placed before $x_j$ in $\ar_1^*$, we have
$$
\myP_{\ar_1\sim \D^1_a} \left[\aris{1}{k} = \mS^i\right] > \myP_{\ar_1\sim \D^1_a} \left[\aris{1}{k} = \mS^j\right]\,.
$$
By applying the rearrangement inequality, then \ref{ineq:pair_sl_sr} holds.

\paragraph{Case II: Picking $x_i$ or $x_j$.}
We next consider the change of the probabilities of picking $x_i$ or $x_j$.
First, we show that the probability of picking $x_i$ decreases after the algorithm places $x_i$ after $x_j$.
Similarly, 
\begin{align*}
\myP[x_C^1 = x_i] &= \sum_{\mS: x_i \in \mS } \myP\left[\aris{1}{k}= \mS \right] \cdot \myP_{\hr \sim \D_h}\left[x_i \succ_{\hr} \mS \right], \\
\myP[x_C^2 = x_i] &= \sum_{\mS: x_i \in \mS} \myP\left[\aris{2}{k} = \mS \right] \cdot \myP_{\hr \sim \D_h}\left[x_i \succ_{\hr} \mS \right].
\end{align*}
Next, we still compare the two probabilities by term.
Consider the terms where the set $\mS$ contains both $x_i$ and $x_j$.
The first terms $\myP[\aris{1}{k} = \mS]$ and $\myP[\aris{2}{k} = \mS]$ are the same by identical logic used above.
Next, consider the other terms where $\mS$ does not contain $x_j$.
Let $\mS' = \mS\setminus \{x_i\} \cup \{x_j\}$.
By a similar reason, $\myP[\aris{2}{k} = \mS] = \myP[\aris{1}{k} = \mS']$.
In addition, as $x_i$ is placed before $x_j$ in $\ar_1^*$, then by \Cref{lem:mono_presenting_items}, $\myP[\aris{1}{k} = \mS] > \myP[\ars{k} = \mS']$, which further concludes the first inequality.
The remaining comparison of the probabilities that the human selects $x_j$ under the two algorithms follows by symmetry.
\end{proof}

\ThmBenefit*
\begin{proof}
Consider the change in the probability of the human picking each item.
By the above lemma, after swapping $x_i$ to a later position in algorithm's ground-truth, only the probability of the human picking item $x_i$ decreases, while the probability of picking any other item increases.
Therefore, if both $x_i$ and $x_j$ are zero-valued to the human, then such a swap leads to an increase in the probability of picking any valuable item, which in turn increases human's utility.
For the second bullet, if $x_i$ is the most valuable item to the human, then the swap causes a loss in picking the most valuable item, which further leads to a decrease in human's utility.
\end{proof}

\CoroTopItem*
\begin{proof}
By the first bullet of \Cref{thm:benefit_inv1}, swapping any pair of zero-valued items will increase the human's utility.
Therefore, for the algorithm that maximizes human's utility, all items other than $x_1$ should be ranked in reverse order.
In addition, by the second bullet, ranking human's top item $x_1$ in any place other than the first position will decrease human's utility.
Thus, the optimal algorithm should rank $x_1$ first.
Using a symmetric argument, we can also prove that the ranking $(x_2, x_3, \ldots, x_m, x_1)$ yields the least benefit.

Notably, this result extends to the setting where the human has multiple top-valued items. Suppose the human assigns a positive value only to her top \( d \) items, which are equally valued. In this case, the algorithm's arrangement \( \ar^* \) that maximizes the human's expected utility places the top \( d \) items first, followed by the remaining items in reverse order; that is, $\ar^* = (x_1, \dots, x_d, x_m, \dots, x_{d+1})$.
\end{proof}

\subsection{Extensions of \Cref{thm:benefit_inv1}}
\label{app:extension_plackett_luce}

\lemLessUtilityMallows*
\lemLessUtility*
\begin{proof}
According to the computation of \Cref{lem:prob_of_picking_item_i1}, for every $r\neq i, j$, we have 
\begin{align*}
\myP[x_C^2 = x_r] - \myP[x_C^1 = x_r] = \left(\myP[\pi_1[:2] = \{x_i, x_r\}] - \myP[\pi_1[:2] = \{x_j, x_r\}]\right)\cdot \left(\myP[x_r \succ_{\rho} x_j] - \myP[x_r \succ_{\rho} x_i]\right)\ge 0
\end{align*}
Meanwhile, we have 
\begin{align*}
\myP[x_C^2 = x_j] - \myP[x_C^1 = x_j] = \sum_{r\neq i, j}\left(\myP[\pi_1[:2] = \{x_i, x_r\}] - \myP[\pi_1[:2] = \{x_j, x_r\}]\right)\cdot \myP[x_j \succ_{\rho} x_r] \ge 0\\
\myP[x_C^2 = x_i] - \myP[x_C^1 = x_i] = \sum_{r\neq i, j}\left(\myP[\pi_1[:2] = \{x_i, x_r\}] - \myP[\pi_1[:2] = \{x_j, x_r\}]\right)\cdot \myP[x_i \succ_{\rho} x_r] \le 0
\end{align*}
Let $\psi(i, j, r)$ be $\psi(i, j, r) = \myP[\pi_1[:2] = \{x_i, x_r\}] - \myP[\pi_1[:2] = \{x_j, x_r\}]$.
Therefore, the expected change in human's utility is given by
\begin{align*}
&\sum_{r\neq i, j} v_r \cdot (\myP[x_C^2 = x_r] - \myP[x_C^1 = x_r]) + v_j \cdot (\myP[x_C^2 = x_j] - \myP[x_C^1 = x_j]) + v_i \cdot (\myP[x_C^2 = x_i] - \myP[x_C^1 = x_i])  \\
= & \sum_{r\neq i, j} v_r \cdot  \psi(i, j, r)\cdot \left(\myP[x_r \succ_{\rho} x_j] - \myP[x_r \succ_{\rho} x_i]\right) + \sum_{r\neq i, j} v_j \cdot \psi(i, j, r)\cdot \myP[x_j \succ_{\rho} x_r] -  \sum_{r\neq i, j} v_i \cdot \psi(i, j, r)\cdot \myP[x_i \succ_{\rho} x_r]  \\
= & \sum_{r\neq i, j} v_r \cdot  \psi(i, j, r)\cdot \left(\myP[x_i \succ_{\rho} x_r] - \myP[x_j \succ_{\rho} x_r]\right) +  \sum_{r\neq i, j} v_j \cdot \psi(i, j, r)\cdot \myP[x_j \succ_{\rho} x_r] -  \sum_{r\neq i, j} v_i \cdot \psi(i, j, r)\cdot \myP[x_i \succ_{\rho} x_r] \\
= & \sum_{r\neq i, j}   \psi(i, j, r)\cdot \left((v_r - v_i) \cdot (\myP[x_i \succ_{\rho} x_r]) - (v_r - v_j) \cdot (\myP[x_j \succ_{\rho} x_r]) \right) 
\end{align*}

Next, we show that the above summation is always negative when the preconditions of the two models are satisfied.
Notice that, when $i \le r \le j$, by assumption, we have $v_i \ge v_r \ge v_j$.
Hence, the inner term is always negative (as $v_i \neq v_j$).
In addition, when $r \ge j$, by \Cref{lem:CompMono}, we know $\myP[x_i \succ_\rho x_r] \ge \myP[x_j \succ_\rho x_r]$.
Meanwhile, we can observe that $\abs{x_r - x_i} \ge \abs{x_r - x_j}$.
Hence, the inner term is still negative.
It suffices to prove that every term with $r < i$ is negative.

\paragraph{(a) Plackett-Luce model.} Let $\delta = v_r - v_i$ and $\Delta = v_i - v_j$.
By the definition of the Plackett-Luce model, the inner term can be rewritten as 
\begin{align*}
(v_r - v_i)\cdot \frac{\exp(v_i/\beta)}{\exp(v_r/\beta) + \exp(v_i/\beta)} - (v_r - v_j)\cdot \frac{\exp(v_j/\beta)}{\exp(v_r/\beta) + \exp(v_j/\beta)} = \frac{\delta}{\exp(\delta/\beta) + 1} - \frac{\delta +\Delta}{\exp((\delta +\Delta) / \beta) + 1}\,.
\end{align*}
Let function $f(x) = \frac{x}{\exp(x/\beta) + 1}$.
Since $f'(x) = \frac{(1-x/\beta)\exp(x/\beta) +1 }{(\exp(x/\beta) + 1)^2}$, then $f(x)$ is increasing when $x\le 1.278$.
Therefore, the above term is always negative since $\delta + \Delta  = v_r - v_j \le v_0 - v_j \le 1.27\beta$.

\paragraph{(b) Mallows model.} Since swapping item $x_i$ and $x_j$ at most increases the number of inversion by $\abs{j-i}$, then $\myP[x_j \succ_\rho x_r] \ge \exp(-\phi_h\cdot (i-j))\cdot \myP[x_i \succ_\rho x_r]$.
Since $v_r \le v_0 \le \frac{\exp(\phi_h\cdot (i-j))}{\exp(\phi_h\cdot (i-j) - 1}v_i  - \frac{1}{\exp(\phi_h\cdot (i-j)) - 1}v_j$, we have $v_r - v_i \le \exp(-\phi_h\cdot (i-j))\cdot (v_r - v_j)$, which means that the inner term is always non-positive and concludes the proof.
\end{proof}

\lemMoreUtilityMallows*
\lemMoreUtility*
\begin{proof}
Let $i'$ be the index satisfying the given condition in each of the two models, respectively.
Then the expected change is at least
\begin{align*}
& \sum_{r\neq i, j} v_r \cdot  \psi(i, j, r)\cdot \left(\myP[x_r \succ_{\rho} x_j] - \myP[x_r \succ_{\rho} x_i]\right) + \sum_{r\neq i, j} v_j \cdot \psi(i, j, r)\cdot \myP[x_j \succ_{\rho} x_r] -  \sum_{r\neq i, j} v_i \cdot \psi(i, j, r)\cdot \myP[x_i \succ_{\rho} x_r] \\
> & \sum_{r=1}^{i'}  v_r \cdot  \psi(i, j, r)\cdot \left(\myP[x_r \succ_{\rho} x_j] - \myP[x_r \succ_{\rho} x_i]\right)  - \sum_{r\neq i, j} v_i\cdot \psi(i, j, r)\cdot \myP[x_i \succ_{\rho} x_r]
\end{align*}

\paragraph{(a) Plackett-Luce model.} When the human's ranking satisfies the Plackett-Luce model, then the above value is at least
\begin{align*}
&\sum_{r=1}^{i'}  v_r \cdot  \psi(i, j, r)\cdot \left(\frac{\exp(v_r/\beta)}{\exp(v_j/\beta) + \exp(v_r/\beta)} - \frac{\exp(v_r/\beta)}{\exp(v_i/\beta) + \exp(v_r/\beta)}\right)  - \sum_{r\neq i, j} v_i\cdot \psi(i, j, r)\cdot \frac{\exp(v_i/\beta)}{\exp(v_i/\beta) + \exp(v_r/\beta)}  \\
\ge & \sum_{r=1}^{i'}  v_r \cdot  \psi(i, j, r)\cdot \frac{\exp(v_r/\beta)\cdot (\exp(v_i/\beta) - \exp(v_j/\beta))}{(\exp(v_i/\beta) + \exp(v_r/\beta))(\exp(v_j/\beta) + \exp(v_r/\beta))}  - \sum_{r\neq i, j} v_i\cdot \psi(i, j, r)\cdot \frac{\exp(v_i/\beta)}{\exp(v_i/\beta) + \exp(v_r/\beta)} \\
\ge & \sum_{r=1}^{i'}  v_r \cdot  \psi(i, j, r)\cdot \frac{\exp(v_r/\beta)\cdot (\exp(v_i/\beta) - \exp(v_j/\beta))}{(\exp(v_i/\beta) + \exp(v_r/\beta))^2}  - \sum_{r\neq i, j} v_i\cdot \psi(i, j, r)\cdot \frac{\exp(v_i/\beta)}{\exp(v_i/\beta) + \exp(v_r/\beta)} \\
\ge & v_{i'}\cdot \sum_{r=1}^{i'} \psi(i, j, r)\cdot \frac{\exp(v_r/\beta)\cdot (\exp(v_i/\beta) - \exp(v_j/\beta))}{(\exp(v_i/\beta) + \exp(v_r/\beta))^2} - v_i \cdot \sum_{r\neq i, j}\psi(i, j, r)\cdot \frac{\exp(v_i/\beta)}{\exp(v_i/\beta) + \exp(v_r/\beta)}\\
\ge & v_{i'}\cdot \sum_{r=1}^{i'} \psi(i, j, r)\cdot \frac{ \exp(v_i/\beta) - \exp(v_j/\beta)}{2(\exp(v_i/\beta) + \exp(v_r/\beta))} - v_i \cdot \sum_{r\neq i, j}\psi(i, j, r)\cdot \frac{\exp(v_i/\beta)}{\exp(v_i/\beta) + \exp(v_r/\beta)} \\
\ge & 0\,.
\end{align*}

\paragraph{(b) Mallows model.} When the human's ranking satisfies Mallows model, we can notice that for any $r\le i'$, 
\begin{align*}
\myP[x_r \succ_{\rho} x_j] - \myP[x_r \succ_{\rho} x_i] & \ge \frac{1}{1 - \exp(-\phi_h\cdot (r-j+1))} - \frac{1}{1 - \exp(-\phi_h\cdot (r-i))} \\
&\ge \frac{\exp(-\phi_h\cdot (r- j+1)) - \exp(-\phi_h\cdot (r-i))}{(1 - \exp(-\phi_h\cdot (r-i)))\cdot (1 - \exp(-\phi_h\cdot (r-j+1)))}\\ 
&\ge \exp(-\phi_h\cdot(r-j+1))\cdot (1-\exp(-\phi_h\cdot (i-j-1))
\end{align*}
Therefore, the expected change is at least
\begin{align*}
& \sum_{r=1}^{i'}  v_r \cdot \psi(i, j, r) \cdot \exp(-\phi_h\cdot(r-j+1))\cdot (1-\exp(-\phi_h\cdot (i-j-1)) - \sum_{r\neq i, j} v_i\cdot \psi(i, j, r)\cdot \myP[x_i \succ_{\rho} x_r] \\
\ge &v_{i'}\cdot  (1-\exp(-\phi_h\cdot (i-j-1))\cdot \sum_{r=1}^{i'} \psi(i, j, r)\cdot \exp(-\phi_h\cdot(r-j+1))- v_i \sum_{r\neq i, j} \psi(i, j, r) \\ 
\ge & 0, \tag{By assumption of $v_{i'}$ and $v_i$}
\end{align*}
which concludes the lemma.
\end{proof}

\section{Omitted Proofs in Section 5}
\label{app:strategic_algo}

\subsection{\NP-Hardness of Welfare-Maximization}
\label{app:strategic_algo_welfare_max}
\thmSWNP*
\begin{proof}
As discussed, the optimal solution is noiseless.
Next, we reduce from the Independent Set problem, where the input is a graph instance $G=(V, E)$ with $n$ vertices and $m$ edges, and an integer $k$.
The output is \YES if the size of the maximum independent set of graph $G$ is at least $k$ and \NO otherwise.
We construct the human-algorithm collaboration instance in the following way.

Fix $\epsilon >0$ to be a sufficiently small constant such that $(1-\epsilon)^{k-1} > 1- 1/(2k)$ and $\epsilon <1/4$.
Construct $m$ items $x_1, \ldots, x_m$. 
Denote the $n$ vertices as $v_1, \ldots, v_n$.
For each $i\in [n]$, we construct a ground-truth ranking $\hr_i^*$ as follows: place $x_i$ as the first item in $\hr_i^*$.
Iterate through all the other vertices from $v_1$ to $v_n$ and insert $x_j$ after $x_i$ in $\hr_i^*$ if vertex $v_i$ is adjacent to vertex $v_j$ in the input graph instance $G$.
Otherwise, put $v_j$ to the end of $\hr_i^*$.
For the remaining items $x_{n+1}, \ldots, x_m$, put them in the middle of the inserted items.
Set $p_i = 1/n$ for each $i\in [n]$ and $v_1 = 1$, $v_i =0$ for $i\ge 2$.

\begin{figure}[ht]
\centering
\begin{tikzpicture}[
    redroundnode/.style={circle, draw=black, thick, minimum size=0.5cm},
    orangeroundnode/.style={circle,draw=black, thick, minimum size=0.5cm},
    blackroundnode/.style={circle, draw=black,,  thick, minimum size=0.5cm},
    ]
\node[scale=0.7] at (9, 1) {\Large $\hr^*_j: \quad x_j \succ x_i \succ x_{n+1} \succ \cdots \succ x_m \succ x_j$};
\node[scale=0.7] at (9, 0) {\Large $\hr^*_i: \quad x_i \succ x_j \succ x_{n+1} \succ \cdots \succ x_m \succ x_k$};
\node[scale=0.7] at (9, -1) {\Large $\hr^*_k: \quad x_k \succ x_i \succ x_{n+1} \succ \cdots \succ x_m \succ x_j$};

\node[] at (9, -2) {\text{Constructed rankings $\hr_i^*, \hr_j^*$, and $\hr_k^*$}};

\node[blackroundnode, scale=0.8] (vi) at (2, 0) {$v_i$};
\node[redroundnode, scale=0.8] (vj) at (3, 1) {$v_j$};
\node[orangeroundnode, scale=0.8] (vk) at (3, -1) {$v_k$};

\node[] at (3, -2) {\text{Input graph $G$}};
\draw[-, thick] (vi) -- (vj); 
\draw[-, thick] (vi) -- (vk); 
\end{tikzpicture}
\end{figure}
We choose $\phi_h$ and $m$ so that in a Mallows model with $m$ items and accuracy parameter $\phi_h$, given two items, the probability of ranking before the other is at most than $1/2+ \epsilon$ when the difference of the indices is smaller than $n$ and is larger than $1-\epsilon$ when it is larger than $m-n$, which can be achieved by \Cref{lem:set_constant}.

Next, we first show that if the input graph instance is a \YES instance, then the optimal expected social welfare is at least $(1-\epsilon)^k\cdot k/n$.
Suppose the vertex index set of the maximum independent set of $G$ is $I$ with $\abs{I} \ge k$.
Put $k$ of the independent set as the first $k$ items of $\ar^*$ and the other $m-k$ items in an arbitrary order. 
Hence, specifying $\mS = \arss{k}$, for every human $i\in I$, the probability of picking her best item is at least
\begin{align*}
\myP[x_C^i = x_i \mid \mS] & = \myP[x_i \succ_{\hr_i} \mS] = \myP\Big[\wedge_{j\in I}(x_i \succ_{\hr_i} x_j)\Big] \ge \prod_{j\in I} \myP\left[x_i \succ_{\hr_i} x_j\right] \ge (1-\epsilon)^{k-1}\,.
\end{align*}
Therefore, the expected utility of the human with ground-truth $\hr_i^*$, $\myE[u(x_J\mid \hr_i^*)]$ is at least $(1-\epsilon)^{k-1}$.
As the human has a probability of $p_i = 1/n$ of having ground-truth $\hr_i^*$ for every $i\in I$, the expected social welfare is at least $(1-\epsilon)^{k-1}\cdot k/n$.

Next, we consider when the input instance is \NO.
Denote the $k$ items of $\ar^*$ by $\mS$.
As the size of $\mS$ is $k$, there are $n-k$ humans whose best item is not picked.
Hence, their utilities have a probability of at least $(1-\epsilon)$ being $0$.
In addition, as the input instance is \NO instance, consider the vertices corresponding to $\mS$.
There must be two adjacent vertices included at the same time.
Hence, there will be two humans (say $i, j$) whose best item is within the top $n$ items of the other one's ground-truth ranking.
By the assumption of $\phi_h$, the two humans only have the probability of no more than $1/2+\epsilon$ to pick their best item.
Therefore, the expected social welfare is at most 
\begin{align*}
&\frac1n\cdot \Big(\myE[u_i(x_C^i)] + \myE[u_j(x_C^j)] + \sum_{t\in I\setminus \{i, j\}}\myE[u_t(x_C^t)] + \sum_{t\notin I}\myE[u_t(x_C^t)]\Big)  \\
& \le \frac{1}n\cdot \left(2\cdot \Big(\frac12  + \epsilon\Big) + (k-2)\cdot 1 \right)\\
& < \frac1n\cdot \left(k - 1 + 2 \epsilon\right) < \frac{1}n\cdot k \cdot (1-\epsilon)^{k-1} \tag{as $(1-\epsilon)^{k-1} > 1- 1/(2k)$ and $\epsilon < 1/4$},
\end{align*}
which is smaller than the expected social welfare under the \YES instance.
Therefore, the problem of finding a welfare-maximizing strategy is \NP-hard for the algorithm.
\end{proof}

\begin{observation}
We observe that for a fixed arrangement $\ar^*$, noiselessness is not always optimal in terms of the social welfare.
Consider the following example.
The algorithm's ground-truth is fixed as $\ar^*=(x_1, x_2, x_3)$, but items $x_1$ and $x_2$ only have positive value to a proportion of $1\%$ of people.

If the algorithm is noiseless and only presents two items to the humans, it causes $99\%$ of people to get no benefit, which results in worse social welfare than a noisy algorithm that would randomize over items that are presented.
\end{observation}

\subsection{Mixed Integer Linear Programming Formulation}
\label{app:mip_formulation}
Finally, we adapt the MIP of \cite{desir2016assortment} and formulate the problem of finding the welfare-maximizing strategy as an MIP.
\cite{desir2016assortment} studies the \emph{assortment optimization} problem under the Mallows model, where each item is associated with a nonnegative value $r_i$ and the goal is to pick a subset of items $\mS$ to maximize the weighted sum of choice probabilities.
Formally, the problem can be expressed as follows:
$$
\max_{\mS \subseteq [m]} \left(\sum_{x_i \in \mS} r_i \cdot \myP[x_i \succ \mS]\right)\,.
$$
\cite{desir2016assortment} shows the choice probability of $x_i$ being the first among all items in $\mS$ can be calculated by a dynamic programming algorithm.
However, it is worth noting that the algorithm relies on prior knowledge that some item is known to be included in $\mS$, which is further used as a ``guess'' of the MIP formulation.
We next provide a dynamic programming of the choice probability with marginal change to the original algorithm of \cite{desir2016assortment} without any prior knowledge.

Let $W(i, s, t)$ be the probability of item $x_i$ being chosen at position $s$ after $t$ steps of repeatedly insertion, where $i\in [m]$, $s\in [t]$, and $t\in [m]$.
The following algorithm assumes the ground-truth ranking $\pi^* = (x_1, \ldots, x_m)$ without loss of generality.
We use the same notations for constants $p_{t+1, s}$, $\gamma_{t+1, s}$ as in \cite{desir2016assortment}.
Probability $p_{t+1, s}$ is the probability of inserting item $x_{t+1}$ at position $s$ into a partial permutation of $x_1, \ldots, x_t$.
Also, $\gamma_{t,s} = \sum_{\ell=1}^s p_{t,\ell}$.
Both sets of constants have closed-form expressions once the ground-truth ranking and the accuracy parameter are given.
\begin{algorithm}[t]
\caption{Dynamic programming for choice probability without prior knowledge}
\SetKwProg{Fn}{Function}{:}{}
\Fn{{\tt ChoiceProb}$(\mS, x_i)$}{
\KwIn{a subset $\mS$ of items $[m]$;}
\KwOut{Probability of item $x_i$ being the first among all items in $\mS$;}
\For{$t=1, \ldots, m$} {
    \For{each $1\le i\le k-1$ and $s=1, \ldots, t$} {
        $W(i, s, t)\leftarrow (1 - \gamma_{t, s})\cdot W(i, s, t-1) + \mathbf{1}[x_t \notin \mS]\cdot \gamma_{t, s-1}\cdot W(i, s-1, t-1)$\;
    }
    \For{each $s=1, \ldots, t$} {
        $W(t, s, t) \leftarrow \mathbf{1}[x_t \in \mS]\cdot p_{t, s}\cdot \left(\sum_{i\le t-1}\sum_{\ell = s}^m W(i, \ell, t-1) + \prod_{i=1}^{t-1}\mathbf{1}[x_i\notin \mS]\right)$\;
    }
}
\Return{$\sum_{s=1}^m W(i, s, m)$}\;
}
\label{alg:dpChoice}
\end{algorithm}

\begin{lemma}
Algorithm~\ref{alg:dpChoice} computes the choice probability of item $x_i$ being the first among $\mS$.    
\end{lemma}
\begin{proof}
Compared to the original dynamic programming of \cite{desir2016assortment}, we do not require a precondition that an item $x_i$ needs to be included in $\mS$.
Hence, we update the state space starting from step $t=1$ rather than $t=2$.
As a result, if $x_t\in \mS$, at the $t$-th step, there are two possible cases for item $x_t$ being selected at position $s$: 1) none of items $x_1, \ldots, x_{t-1}$ is included in $\mS$; therefore, item $x_t$ is selected at position $s$ with probability $p_{t, s}$; 2) at least one of items $x_1, \ldots, x_{t-1}$ is included in $\mS$; therefore, item $x_t$ is chosen only if all items of $x_1, \ldots, x_{t-1}$ that are included in $\mS$ are ranked after $x_t$ and $x_t$ is inserted at position $s$ (the same argument of \cite{desir2016assortment}).
Otherwise, if $x_t\notin \mS$, the same argument of \cite{desir2016assortment} still applies for updating the probabilities of choosing other items.
\end{proof}

\paragraph{MIP Formulation.}Based on the above dynamic programming, we first give a simplified MIP formulation of the assortment problem in Fig.~\ref{fig:mip_formulation}.
For notational simplicity, we reuse $\sigma$ to denote the ground-truth ranking.
Note that, we no longer force $\sigma$ to be $\{x_1, \ldots, x_m\}$.
Denote by $\sigma(i)$ the index of the $i$-th item.
It is worth noting, the new MIP does not rely on the guess of knowing some item included in $\mS$.
Therefore, it is no longer needed to enumerate all the possible guesses, which further reduces the running time by a factor of $O(m)$.

Based on the MIP, we then formulate the problem of finding the welfare-maximizing algorithm as an MIP as follows.
We replicate the constraint of choice probability $n$ times -- once for each of the $n$ Mallows distributions corresponding to the $n$ types of humans.
Notably, as each type of human has a different ground-truth ranking, $\sigma$ is set accordingly in each case.
In addition, the objective function is changed accordingly by summing the utilities of all humans.

\begin{figure}[h]
\begin{framed}
\textbf{MIP for the Assortment Problem}
\begin{align*}
\max_{\bm{x},\bm{W}} \quad & \sum_{i=1}^m r_i\cdot \sum_{s=1}^m W(i, s, m)  & \\
\text{subject to} \quad & W(i, s, t) = (1 - \gamma_{t, s})\cdot W(i, s, t-1) + y_{i, s, t} && \forall i, s, t\\
& y_{i, s, t} \le \gamma_{t, s-1}\cdot W(i, s-1, t-1) && \forall j, s, t \\
& W(t, s, t) = z_{s, t} && \forall s, t \\
& z_{s, t} \le p_{t, s} \cdot \left(\sum_{i = 1}^{t-1}\sum_{\ell=s}^m W(i, \ell, t-1) + q_{t-1}\right) && \forall s, t\\
& 0 \le q_t \le 1 - x_{\sigma(i)} && \forall i \le t \\ 
& q_t \ge 1 - \sum_{i=1}^{t} x_{\sigma(i)} && \forall t \\
& 0 \le z_{s, t} \le p_{t, s}\cdot x_{\sigma(t)} && \forall s, t\\
& x_i \in \{0,1\} && \forall i\\
& \sum_{i=1}^m x_i \le k
\end{align*}
\begin{itemize}[leftmargin=0.5cm]
    \item $\sigma$: the ground-truth ranking of the Mallows model;
    \item $r_i$: the value of item $x_i$;
    \item $\sigma(i)$: the index of the $i$-th item in the ground-truth ranking $\sigma$;
    \item $x_i$: indicator variable of whether item $x_i$ is included in the assortment $\mS$;
    \item $W(i, s, t)$: the probability of the $i$-th item being chosen at position $s$ after $t$ steps of repeatedly insertion;
    \item $p_{t+1, s}$ is the probability of inserting item the $(t+1)$-th item at position $s$ into a partial permutation of the top $t$ items;
    \item $\gamma_{t,s}$: $\sum_{\ell=1}^sp_{t, \ell}$.
\end{itemize}

\textbf{MIP for the Welfare Maximization Problem}

\begin{align*}
\max_{\bm{x},\bm{W}^{(1)}, \ldots,\bm{W}^{(n)}} \quad & \sum_{i=1}^n p_i\cdot \sum_{j=1}^m v_j\cdot \sum_{s=1}^m W^{(i)}(i, s, m)  & \\
\text{subject to} \quad & \texttt{choiceProbConstraint}(i, \bm{x}, \bm{W}^{(i)}) && \forall i \in [n]
\end{align*}
\begin{itemize}[leftmargin=0.5cm]
    \item $\texttt{choiceProbConstraint}(i, \bm{x}, \bm{W}^{(i)})$ refers to the constraint of choice probability between $\bm{x}$ and $\bm{W}^{(i)}$ under the $i$-th Mallows distribution. 
\end{itemize}

\end{framed}
\caption{MIP formulation of the welfare-maximizing strategy.}
\label{fig:mip_formulation}
\end{figure}

\subsection{Uplift under Special Cases}
\label{app:uplift_and_comp}
\lemAligedTopItems*
\begin{proof}
We first show the lemma holds for a special case: $\phi_a = \phi_h = \phi$.
Unfold the expectation of the two sides of the inequality of uplift.
The original inequality is equivalent to the following inequality
\begin{align}\label{eqn:collaborate_benefit_ineq1}
\sum_{j=1}^m v_j \cdot \myP[x_C = x_j] > \sum_{j=1}^m v_j \cdot \myP[x_H = x_j]\,.
\end{align}
To show the above inequality, we apply Karamata's inequality to prove the following proposition, which essentially states that, for any $i\in [T]$, the probability of the joint system picking an item of $x_1, \ldots, x_i$ is weakly larger than the probability of the human acting alone. 
\begin{proposition}
\label{prop:prob_of_first_k_being_selected}
For any $i\in [T]$, we have $\sum_{j=1}^i \myP[x_C = x_j] > \sum_{j=1}^i \myP[x_H = x_j]$. 
\end{proposition}
\begin{proof}
According to the law of conditional probability, we can expand the above probability conditioned on the first $k$ items of human's permutation as follows,
\begin{align*}
\sum_{j=1}^i \myP[x_H = x_j] & =
\sum_{j=1}^i \sum_{\mS: \abs{\mS} = k} \myP[x_H = x_j \mid \hrs{k} = \mS] \cdot \myP[\hrs{k} = \mS]
\end{align*}
By \Cref{lem:MallowsFirstk}, the probability of $\myP[\hrs{k} = \mS]$ only depends on $\phi$ and the sum of indices of items of $\mS$.
As $\hr^*$ an $\ar^*$ have the same set of top $T$ items, we can see $\myP[\hrs{k} = \mS] = \myP[\ars{k} = \mS]$ and the above sum can be rewritten as
\begin{align*}
\sum_{j=1}^i \myP[x_H = x_j]  =
\sum_{S: \abs{S} = k}\myP[\ars{k} = S] \cdot  \sum_{j=1}^i \myP[x_H = x_j \mid \hrs{k} = S]\,.
\end{align*}    
Fix a set $S$ of size $k$.
Notice that, when $x_j \notin S$, $\myP[x_H = x_j \mid \hrs{k} = \mS] = 0$.
Next, we consider the sum of the conditional probabilities for $x_j \in S$.
We first show that, conditioned on $\hrs{k} = \mS$, the partial ordering of the first $k$ items (denoted by $\hr_{|\mS}$) forms a Mallows distribution of items of $\mS$.
Observe that the number of inversions between $S$ and $M\setminus S$ for every permutation $\hr$ such that $\hrs{k} = \mS$ since the two sets of items are isolated.
Hence, among all the permutations such that $\hrs{k} = \mS$, the probability of each of them is proportional to $\exp(-\phi)$ raised to the power of the number of inversions inside $\mS$, which implies that the partial ranking $\pi_{|S}$ forms a Mallows distribution of items of $\mS$.
Denote the $s$ items by $\mS = \{x_{i(1)}, \ldots, x_{i(k)}\}$ with $i(1) < \ldots <  i(k)$.
By the known property of Mallows model~\cite[Lemma 2.5]{awasthi2014learning}, the probability of item $x_{i(j)}$ being the first is equal to $\exp(-\phi\cdot (j-1))/ Z_k(\phi)$.
Summing it up for $j$ from $1$ to $i$, we get the probability of $x_H$ being one of $x_1$ to $x_i$ is equal to $Z_s(\phi)/Z_k(\phi)$, where $s$ is the size of the intersection of $\{x_1, \ldots, x_i\}$ and $\mS$.
Therefore, the probability of the human picking one of the first $i$ items is equal to 
\begin{align*}
\sum_{j=1}^i \myP[x_H = x_j]  =
\sum_{S: \abs{S} = k}\myP[\ars{k} = \mS] \cdot \frac{Z_{s}(\phi)}{Z_k(\phi)},\quad \text{with } s = \abs{S\cap \{x_1, \ldots, x_i\}}.
\end{align*}    
On the other side, by the definition of the pick-and-choose, the joint system picks the human's favorite item among the presented $k$ items of the algorithm.
Hence, we have
\begin{align}
\label{eq:aligned_joint_system}
\sum_{j=1}^i \myP[x_C = x_j] = \sum_{\mS: \abs{\mS} = k}\myP[\ars{k} = \mS] \cdot  \sum_{j=1}^i \myP[x_j \succ_{\hr} \mS]\,.
\end{align}
In \Cref{lem:general_lower_bound_for_comparison}, we show that $\sum_{j=1}^i \myP[x_j \succ_{\hr} S]\ge Z_{s}(\phi) / Z_k(\phi)$.
It immediately implies that $\sum_{j=1}^i \myP[x_C = x_j] \ge \sum_{j=1}^i \myP[x_H = x_j]$ and concludes the proof.
\end{proof}
Using Karamata's inequality and \Cref{prop:prob_of_first_k_being_selected}, we conclude that the theorem holds.
\end{proof}

\subsection{Small Noiseness Helps Uplift}
\label{app:noiseness_help_uplift}

\lemPostiveTemp*
\begin{proof}
Suppose there are three types of humans, and an algorithm is presenting $k=2$ items to the human.
The ground-truth rankings of the three types of humans are respectively $\hr_1^* = (x_1, x_2, \cdots, x_m)$, $\hr_2^* = (x_2, x_3, \cdots, x_m)$, and $\hr_3^* = (x_3, x_1, \cdots, x_m)$.
Consider the top-item recovery setting: $v_1 = 1$ and $v_i = 0$ for $i\ge 2$.
The accuracy parameter of the human $\phi_h$ is assumed to be sufficiently small such that every human is almost equally likely to pick any item by themselves.

If the algorithm sets the response to be noiseless, then the presented two items will be deterministic, which means only two of $\{x_1, x_2, x_3\}$ can be presented.
As a result, at least one of the three types of humans is unable to pick their best item and receives no utility from the collaboration.

Nevertheless, once we add some noise to the response of the algorithm, uplift becomes possible.
Consider setting the ground-truth ranking of the algorithm as $\ar^* = (x_1, x_2,x_3, \ldots, x_m)$ and the accuracy $\phi_a$ such that $\myP[\ari{1} = x_1] = 1/2, \myP[\ari{1} = x_2] = 1/4$, and $\myP[\ari{3} = x_3] =1/8$.  
For every type of human, if her favorite item is included in the presented two items, she has a probability of at least $1/2$ to choose it among the two items.
Therefore, the probability of choosing the best item reaches at least $1/8\cdot 1/2 = 1/16$, which beats the value of acting alone.
\end{proof}

\subsection{\NP-Hardness of Deciding Existence of Uplifting Strategy}
\label{app:np_hard_for_deciding_uplift_existence}
\thmCompGeneral*
\begin{proof}
We reduce from the Vertex Cover problem.
The input of the Vertex Cover problem is a graph instance $G=(V, E)$ along with a number $k$.
The output is \YES if there exists a vertex cover (containing at least one vertex of every edge) of size at most $k$ and \NO otherwise.

We construct a human-AI collaboration model as follows.
By the properties of the Mallows model, for a fixed accuracy $\phi_h$ and ground-truth ranking $\hr^*$, in a Mallows distribution $\hr\sim \D(\hr^*, \phi_h)$, we have 
\begin{align*}
\myP_{\hr\sim \D(\hr^*, \phi_h)}[\hr(1) \in \{\hr^*(1), \hr^*(2)\}] & = \frac{Z_2(\phi_h)}{Z_m(\phi_h)}> \frac{Z_2(\phi_h)}{1/(1-\exp(-\phi))},  \\
\myP_{\hr\sim \D(\hr^*, \phi_h)}[\hr(1) \in \{\hr^*(1), \hr^*(2)\}] &= \frac{Z_2(\phi_h)}{Z_m(\phi_h)} < \frac{Z_2(\phi_h)}{1+ \exp(-\phi_h) +\exp(-\phi_h\cdot 2)}\,.     
\end{align*}
Then set $\epsilon_1$ as a sufficiently small constants such that $1 - \epsilon_1^{1/k}> 2\epsilon_1^{1/(2k)}$ and $\epsilon_1 < 1/4$, and the accuracy parameter $\phi_h$ such that $1- \epsilon_1 < 1 - \exp(-\phi_h\cdot 2)$.
Let $\epsilon_2$ be defined as follows
$$
\epsilon_2  = 1-  \frac{1+\exp(-\phi_h)}{1+ \exp(-\phi_h) +\exp(-\phi_h\cdot 2)} \,.
$$
Further, set $\epsilon_3$ be a sufficiently small constant such that $(1 - \epsilon_3)^n \ge 1 - \epsilon_2$.
Then we fix $\phi_h$ and apply \Cref{lem:set_constant} to find $B$ such that the probability of pairwise comparison is no less than $1-\epsilon_3$ once the index difference is no less than $B$, and also set $B$ larger than $n$.

Construct $n+B$ items in total: $n$ vertex items $x_v$ for any vertex $v\in V$ and $B$ dummy items $y_i$ for $i\in [B]$.
Also, we create one type of human for every edge $e = (u, v)\in E$ with ground-truth ranking 
$$\hr_e^* = (x_u, x_v, y_1, \cdots, y_B, \cdots), \quad \text{for any } (u, v)\in E$$ 
which places items $x_u$ and $x_v$ first, followed immediately by all the dummy items, and finally the remaining vertex items.
Each human only has positive values for the first two items $v_1 = v_2 = 1$ while $v_j = 0$ for $j > 2$.

We next prove the mapping between the two instances.
When the input vertex cover is a \YES instance, there exists a vertex cover $U\subseteq V$ that covers every edge.
Consider the strategy that puts $U$ first in $\ar^*$ and set $\phi_a = +\infty$.
For every human with ground-truth ranking $\hr_e^*$, if both $x_u$ and $x_v$ are included in $U$, then the probability of choosing either $x_u$ or $x_v$ strictly increases.
Otherwise, if only $x_u$ or $x_v$ is included in $U$, as all other vertex items in $U$ are placed to the end of $\hr_e^*$, the probability of pairwise comparison is at least $1-\epsilon_2$.
By the setting of $\epsilon_3$, the probability of ranking $x_u$ or $x_v$ before other items in $U$ is at least $(1-\epsilon_3)^n \ge 1- \epsilon_2$.
Thus, the human still has a higher probability of choosing her favorite item and receives a higher expected utility.
Hence, $(\ar^*, +\infty)$ achieves uplift and the human-AI model is also a \YES instance.

We next consider when the input graph is a \NO instance of the vertex cover problem.
We prove it by contradiction.
Suppose an uplift strategy exists.
We first claim that the algorithm should be ``very random''.
Consider the probability of the algorithm presenting the first $k$ items of its ground-truth ranking.
By \Cref{lem:MallowsFirstk}, the probability of that event is greater than $1/\Pi_{i=1}^{k} Z_m(\phi_a)$.
As the input instance is \NO instance, there must exist an agent whose top two items are not included in the first $k$ items of $\ar^*$.
As a result, one human receives a utility of zero in this case.
Since the utility of the human is at least $1-\epsilon_1$ when acting alone, then 
$$
(1-\exp(-\phi_a))^{k} < \frac{1}{\Pi_{i=1}^{k} 
 Z_m(\phi_a)}
  < \epsilon_1 \implies \exp(-\phi_a) >  1 - \epsilon_1^{1/k}\,.
$$
On the other hand, we prove that, to make sure the algorithm uplifts every human, the algorithm should also be ``very robust'' since it should make sure at least one of the favorite items of every human included in $\mS$ with a high probability.
By \Cref{lem:mono_presenting_items}, the probability of an item ranked in the first $k$ positions reaches maximum when the item is the first one in the ground-truth ranking.
Hence, the maximum probability is $Z_k(\phi_a)/Z_{n+B}(\phi_a)$.
Therefore, 
\begin{align*}
1 - \epsilon_1 & < \myP[\mS \cap \{x_u, x_v\} \neq \emptyset] \le 1 - \myP[\mS \cap \{x_u\} = \emptyset]\cdot  \myP[\mS \cap \{x_v\} = \emptyset] \\
& < 1 - \left(1 - \frac{Z_k(\phi_a)}{Z_{n+B}(\phi_a)}\right)^2 < 1 - \left(1 - \frac{1}{1+\exp(-\phi_a)^k}\right)^2 \tag{as $B > n$}\,.
\end{align*}
Direct calculation implies that
\begin{align*}
1 - \frac{1}{1+\exp(-\phi_a)^k} < \sqrt{\epsilon_1} 
 \implies \exp(-\phi_a)^k < \frac{\sqrt{\epsilon_1}}{1-\sqrt{\epsilon_1}}
\implies \exp(-\phi_a) < \frac{\epsilon_1^{1/(2k)}}{(1-\sqrt{\epsilon_1})^{1/k}} < 2\epsilon_1^{1/(2k)},
\end{align*}
which contradicts the previous inequality $\exp(-\phi_a) > 1 - \epsilon_1^{1/k}$.
Therefore, the human-AI model is also a \NO-instance and it concludes the reduction.

\paragraph{\NP-membership.} Next, we provide a polynomial-time method to validate whether an algorithm with a fixed ground-truth ranking $\ar^*$ and accuracy parameter $\phi_a$ achieves uplift and complementarity when the number of presented items $k$ is constant-bounded.
By the definition of the pick-and-choose, a human picks a particular item $x_i$ only if the algorithm includes it in $\mS$ and the human ranks it before other items of $\mS$.
Notice that the number of possible outcomes of $\mS$ is equal to $n^k$.

For every set of $\mS$ with $k$ items, as shown in \Cref{app:item_loc}, there is a closed-form expression of the probability of $\mS$ being the first $k$ items of $\ar$, which can be computed in $O(m)$ time.
Moreover, Algorithm~\ref{alg:dpChoice} can calculate the probability of an item $x_i$ being first among a shortlist of items $\mS$ in the Mallows model.
Combining them, we can identify the exact utility of every human receiving from the collaboration, which further determines whether uplift is achieved or not.
\end{proof}

\begin{corollary}
The problem of finding the welfare-maximizing strategy remains \NP-hard under the top item recovery setting.
\end{corollary}

\lemTopRecoveryCom*
\begin{proof}
Uplift stems from the algorithm aiding people in eliminating valueless items.
When a human acts alone, she must choose her most valuable item (the ``superstar'') from the entire item set.
As every valuable item remains included in $\mathcal{M}_0$, each human can choose her favorite item among a shorter list, which leads to a higher probability than acting alone.
Meanwhile, to maximize the maximin share, the algorithm should include every item of $\mathcal{M}_0$ in the presented item list $\mathcal{S}$.
Otherwise, there must exist one type of human with no chance to pick her best item, and the minimum utility will be equal to $0$, which is smaller than always presenting $\mathcal{M}_0$.
Furthermore, for every item that is not favored by any human, erasing it from $\mS$ only improves the probability of every human picking her best item, as it causes less misleading.
Thus, the strategy also maximizes the expected utilitarian social welfare among all noiseless strategies.
\end{proof}

\section{Table of Section 6}\label{app:tablesec6}
Expanding on the numerical study presented in \Cref{sec:empirical}, the following table lists the most common sushi preference rankings among the 120 observed rankings. Each row reports a specific ranking, the number of participants (out of 5000) who expressed that preference, and the cumulative fraction of the total population accounted for up to that row. Notably, the top 33 most frequent rankings together represent the preferences of half of the participants.
\begin{table}[htbp]
\centering
\caption{Population distribution and cumulative fractions for sushi rankings.}
\label{tab:tablesec6}
\begin{tabular}{|c|c|c|}
\hline
\textbf{Population} & \textbf{Ranking} & \textbf{Cumulative Fraction of Total Population} \\ 
\hline
126 & [4, 5, 2, 1, 3] & 0.0252 \\
120 & [4, 5, 1, 2, 3] & 0.0492 \\
119 & [4, 5, 2, 3, 1] & 0.0730 \\
119 & [2, 1, 3, 5, 4] & 0.0968 \\
119 & [2, 3, 1, 5, 4] & 0.1206 \\
115 & [1, 2, 3, 5, 4] & 0.1436 \\
94 & [5, 2, 3, 1, 4] & 0.1624 \\
81 & [5, 4, 1, 2, 3] & 0.1786 \\
78 & [1, 4, 5, 2, 3] & 0.1942 \\
77 & [4, 1, 5, 2, 3] & 0.2096 \\
75 & [2, 3, 1, 4, 5] & 0.2246 \\
74 & [4, 2, 5, 3, 1] & 0.2394 \\
73 & [3, 2, 1, 5, 4] & 0.2540 \\
71 & [5, 2, 1, 3, 4] & 0.2682 \\
70 & [5, 4, 2, 1, 3] & 0.2822 \\
67 & [2, 5, 3, 1, 4] & 0.2956 \\
66 & [2, 5, 1, 3, 4] & 0.3088 \\
66 & [1, 2, 3, 4, 5] & 0.3220 \\
65 & [5, 4, 2, 3, 1] & 0.3350 \\
64 & [2, 5, 4, 3, 1] & 0.3478 \\
64 & [4, 5, 1, 3, 2] & 0.3606 \\
64 & [2, 3, 5, 1, 4] & 0.3734 \\
63 & [4, 5, 3, 2, 1] & 0.3860 \\
63 & [5, 2, 3, 4, 1] & 0.3986 \\
62 & [2, 1, 3, 4, 5] & 0.4110 \\
60 & [5, 1, 2, 3, 4] & 0.4230 \\
59 & [4, 1, 5, 3, 2] & 0.4348 \\
56 & [1, 3, 2, 5, 4] & 0.4460 \\
56 & [1, 2, 4, 5, 3] & 0.4572 \\
55 & [4, 2, 5, 1, 3] & 0.4682 \\
54 & [2, 1, 5, 4, 3] & 0.4790 \\
54 & [4, 1, 2, 5, 3] & 0.4898 \\
54 & [3, 1, 2, 5, 4] & 0.5006 \\
\hline
\end{tabular}
\end{table}

\end{document}